\documentclass[journal]{IEEEtran}

\usepackage{amsmath}
\usepackage{amssymb}
\usepackage{amsthm}
\usepackage{slashbox,pict2e}
\usepackage{xcolor}
\usepackage{float}
\usepackage{stfloats}
\usepackage{graphicx}
\usepackage{lineno,hyperref}
\modulolinenumbers[5]
\usepackage{epstopdf}
\usepackage{graphicx,cite,multirow,rotating}
\usepackage{setspace}
\usepackage{pgfplots} 
\usepackage{tikz}
\usepackage{enumitem}   

\usepackage[utf8]{inputenc}
\usepackage{tabularx}
\usepackage{booktabs}

\usetikzlibrary{shapes,arrows}
\usepackage{varwidth}
\usetikzlibrary{shapes.geometric, arrows}
\usetikzlibrary{positioning,arrows}

\usepackage{amsfonts}
\usepackage{yfonts}
\usepackage{psfrag}
\usepackage{pifont}
\usepackage{amsfonts}
\usepackage{bbm}
\usepackage{dsfont}
\usepackage{amssymb}
\usepackage{array}
\usepackage{diagbox} 

\usepackage{color}
\usepackage{amsmath,stackrel,dsfont}
\usepackage{amssymb,hyperref,stmaryrd}
\usepackage{algorithm}
\usepackage{algpseudocode}
\usepackage{pifont}
\modulolinenumbers[5]
\usepackage{amsmath}
\usepackage{epstopdf}
\usepackage{graphicx,cite,multirow,rotating}
\usepackage{setspace}
\usepackage{pgfplots} 
\usepackage{tikz}

\usetikzlibrary{shapes,arrows}
\usepackage{varwidth}
\usetikzlibrary{shapes.geometric, arrows}
\usetikzlibrary{positioning,arrows}
\usepackage{amssymb}
\usepackage{amsfonts}
\usepackage{yfonts}
\usepackage{psfrag}
\usepackage{pifont}
\usepackage{amsfonts}
\usepackage{bbm}
\usepackage{dsfont}

\usepackage{color}
\usepackage{amsmath,stackrel,dsfont}
\usepackage{amssymb,hyperref,stmaryrd}
\usepackage{algorithm}
\usepackage{algpseudocode}
\usepackage{pifont}
\newcommand{\RN}[1]{%
	\textup{\uppercase\expandafter{\romannumeral#1}}%
}

\ifCLASSOPTIONcompsoc
\usepackage[caption=false, font=normalsize, labelfont=sf, textfont=sf]{subfig}
\else
\usepackage[caption=false, font=footnotesize]{subfig}
\fi
\newtheorem{Lemma}{Lemma}
\newtheorem{Proposition}{Proposition}

\begin{document}

\title{Secure SWIPT in the Multiuser STAR-RIS Aided MISO Rate Splitting Downlink}

\author{Hamid~Reza~Hashempour,  Hamed~Bastami, Majid~Moradikia, Seyed~A.~Zekavat, \textit{Senior Member, IEEE}, Hamid~Behroozi, \textit{Member, IEEE},
 Gilberto~Berardinelli, \textit{Senior Member, IEEE}
and A.~Lee~ Swindlehurst, \textit{Fellow, IEEE}
	\thanks{Copyright (c) 2024 IEEE. Personal use of this material is permitted. However, permission to use this material for any other purposes must be obtained from the IEEE by sending a request to pubs-permissions@ieee.org.
		
		Hamid Reza Hashempour and Gilberto Berardinelli's work is partially  funded by the HORIZON JU-SNS-
		2022-STREAM-B-01-03 6G-SHINE project (Grant Agreement
		No.101095738).
		
		Hamid Reza Hashempour and Gilberto Berardinelli are with the Department of Electronic Systems, Aalborg University, Aalborg, Denmark
		(e-mails: \{hrh, gb\}@es.aau.dk).
		
		Hamed Bastami and Hamid Behroozi are with the Department of Electrical Engineering, Sharif University of Technology, Tehran,
		Iran, (e-mails: \{hamed.bastami@ee., behroozi@\}sharif.edu).
		
		Majid Moradikia and Seyed A.Zekavat are with the Department of 
		Physics, Worcester Polytechnic Institute (WPI), Worcester, MA, USA (e-mail:
		mmoradikia@wpi.edu; rezaz@wpi.edu).
		
		A. Lee Swindlehurst is with the Center for Pervasive Communications and
		Computing, Henry Samueli School of Engineering, University of California,
		Irvine, CA, USA 92697 (e-mail: swindle@uci.edu).

	}}

\markboth{IEEE ,~Vol.~?, No.~?,  ~}%
{Shell \MakeLowercase{\textit{et al.}}: Bare Demo of IEEEtran.cls for IEEE Journals}

\maketitle


\begin{abstract}
 Recently, simultaneously
transmitting and reflecting reconfigurable intelligent surfaces (STAR-RISs) have emerged as a novel technology that provides $360^{\circ}$ coverage and new degrees-of-freedom
(DoFs). They are also capable of manipulating signal propagation and  simultaneous wireless information and power transfer (SWIPT). This paper introduces a novel STAR-RIS-aided secure SWIPT system for downlink multiple input single output rate-splitting multiple access (RSMA) networks. The transmitter concurrently communicates with the information receivers (IRs) and sends energy to untrusted energy receivers (UERs). The UERs are also capable of wiretapping the IR streams. We assume that the channel state information (CSI) of the IRs is known at the information transmitter, but only imperfect CSI for the UERs is available at the energy transmitter. By exploiting RSMA, the base station splits the messages of the IRs into common and private parts. The former is encoded into a common stream that can be decoded by all IRs, while the private messages are individually decoded by their respective IRs. 
We find the precoders and STAR-RIS configuration that maximizes the achievable worst-case sum secrecy rate  of the IRs under a total transmit power constraint, a sum energy constraint for the UERs,
and subject to constraints on the transmission and reflection coefficients.
The formulated problem is non-convex and has  intricately coupled variables. To tackle  this challenge, a suboptimal two-step iterative algorithm based on the sequential parametric convex approximation  method is proposed.  Simulations demonstrate that the RSMA-based algorithm implemented with a STAR-RIS enhances both the rate of confidential information transmission and the total spectral efficiency. Furthermore, our method surpasses the performance of both orthogonal multiple access (OMA) and non-OMA (NOMA).
\end{abstract}

\begin{IEEEkeywords}
Rate-Splitting, passive beamforming, reconfigurable intelligent surfaces, simultaneous transmission and reflection, physical layer security.
\end{IEEEkeywords}

\IEEEpeerreviewmaketitle

\section{Introduction}\label{intro}
\IEEEPARstart{D}UE to rapid developments in metasurfaces and their corresponding fabrication technologies, reconfigurable intelligent surfaces (RISs) have emerged as promising candidates for next generation (NG) wireless networks \cite{RIS2,RIS3,RIS4}. RISs are typically constructed as planar
arrays, consisting of many low-cost passive scattering elements. 
Ideally, each reconfigurable element can independently adjust the phase shift and amplitude of the incident wireless signals based on biasing voltages provided by a smart controller. Normally, RISs can only receive and reflect signals in one of the two $180^{\circ}$ half planes defined by its planar structure, and thus an RIS can only serve users on one of its sides \cite{RIS5,RIS6}. However, a new type of RIS, referred to as a simultaneous transmitting/refracting and reflecting RIS (STAR-RIS) \cite{STAR1}, has been introduced to overcome this limitation \cite{STAR2,STAR3}. By employing both electrical-polarization and magnetization currents, STAR-RISs can simultaneously reflect and transmit the incident signals \cite{STAR4}, which leads to full-space $360^{\circ}$ coverage. The transmission and reflection coefficients of the STAR-RIS elements can be configured for flexible deployments.

The inherent broadcast nature of wireless transmissions make them vulnerable to eavesdroppers. Physical layer security (PLS) techniques have been proposed for opportunistically exploiting the random characteristics of fading channels to mitigate the leakage of legitimate information to unintended recipients and thus to enhance the secrecy of wireless communications \cite{PLS2,PLS3,PLS5,PLS6}.
Most of these PLS schemes perform well when accurate channel state information  is available at
the Tx (CSIT) \cite{PLS2,PLS3, PLS5,PLS6}. 

Rate splitting multiple access (RSMA) is a generalized framework combining non-orthogonal-multiple-access (NOMA) \cite{Hanzo-NOMA} and space division multiple access (SDMA), which achieves better multiplexing and spectral efficiency than NOMA and SDMA alone \cite{RSMA}.
In RSMA, each transmitted message is split into a common and
a private part, followed by linear transmit precoding (TPC) at a multi-antenna
Tx and successive interference cancellation (SIC) at each receiver \cite{RSMA}. 
Most of the RSMA literature only considers wireless
information transfer (WIT) via radio-frequency (RF) transmissions. However, RF signals carry not only information
but also energy. Thus, this principle may also be applied for wireless power transfer (WPT) \cite{Hanzo2}. While information receivers (IRs) decode information, energy receivers (ERs) are capable of harvesting energy from the RF signals \cite{Clerckx}. WIT and WPT are unified under the heading of simultaneous wireless information and power transfer (SWIPT), which simultaneously achieves information and energy transmission to IRs and ERs, respectively \cite{Clerckx2,Zhang}. In downlink multiple input single output (MISO) SWIPT networks, 
the TPC has to strike a tradeoff
between information rate and energy transfer \cite{Clerckx3}. 

In this work, motivated by the benefits of RSMA in SWIPT as well as the compelling
applications of STAR-RIS, we conceive a robust and secure SWIPT-aided  STAR-RIS-assisted RSMA downlink solution, in which the transmitter concurrently sends information to IRs and energy to untrusted energy receivers (UERs) in the transmission and reflection spaces of a STAR-RIS. Security is a critical issue since in principle the UERs are able to wiretap the IR streams. We also assume that the channel state information (CSI) of the IRs is known at the transmitter, but only Imperfect CSI (ICSI) is available  for the UERs at the transmitter. 

\subsection{Related Works}
\subsubsection{Studies on STAR-RIS Assisted Networks} STAR-RIS-aided systems have been widely investigated both in orthogonal multiple access (OMA) and NOMA frameworks. For instance, in \cite{Joint design}, a STAR-RIS assisted NOMA system was proposed for maximizing the sum rate by optimizing the power allocation as well as the active and passive beamforming. The coverage range of STAR-RIS-aided two-user communication networks relying on  both NOMA and OMA were studied in \cite{Coverage}. The authors of \cite{Secrecy Design} proposed a secure STAR-RIS-assisted uplink NOMA solution assuming either full or statistical CSI for the eavesdropper. In \cite{Resource Allocation} resource allocation was designed for a STAR-RIS-assisted multi-carrier communication network. In \cite{STAR aided}, a power minimization problem was formulated for the joint optimization of the active beamformer at the BS and  the passive transmission and
reflection beamformers   at the STAR-RIS.
Three practical operating protocols were considered, namely energy splitting (ES), mode switching (MS), and time switching (TS). The authors of \cite{STAR assisted} explored the performance of STAR-RIS-assisted NOMA networks for transmission over Rician fading channels, assuming that the nearby and distant users respectively receive the reflected and transmitted signals from the STAR-RIS. Both exact and asymptotic expressions of the outage probability were also derived in \cite{STAR assisted}. More recently, in \cite{Sum Rate-STAR}, the reflection and transmission coefficients of a STAR-RIS were optimized for ES and MS protocols to maximize the weighted sum rate of a full-duplex communication system. Finally, STAR-RIS-assisted PLS was introduced in \cite{secrecy MISO} for improving the secrecy performance compared to that achievable by conventional RIS. 

\subsubsection{Studies on Secure RIS-Aided SWIPT} 
Liu \textit{et al.} \cite{Energy Efficiency}, maximized the energy efficiency in a secure RIS-aided SWIPT network by jointly optimizing the transmit beamforming
vectors at the access point (AP), the artificial noise (AN) covariance matrix at
the AP, and the phase shifts at the RIS. In \cite{Enhanced secure}, the secrecy rate was maximized 
in a  heterogeneous SWIPT
network with the aid of multiple RISs by carefully designing
the transmit beamforming vector, the artificial noise vector, and reflecting coefficients under a specific  quality-of-service (QoS) constraint.
An RIS-aided secure multiple-input multiple-output (MIMO)  SWIPT system was investigated in \cite{Intelligent Reflect}, where the transmitter, IRs and ERs are all equipped with multiple
antennas and the energy receiver may be a potential eavesdropper (\textit{Eve}). Joint beamforming optimization was used for maximizing the secrecy rate. In \cite{Joint Active}
an RIS-assisted SWIPT MISO network was considered. Specifically, several RISs were utilized to assist
in  information/power transfer from the AP to multiple single-antenna
IRs and ERs.
The goal was to minimize the transmit power at the AP by jointly optimizing the coefficients of the RISs under specific QoS constraints for the users.

\subsection{Motivations and Contributions}
Although some of the above-mentioned research investigates  STAR-RIS systems, PLS is not considered in \cite{Joint design}-\cite{Sum Rate-STAR}. Only \cite{secrecy MISO} introduces a STAR-RIS aided secure network. On the other hand, most recent work on secure SWIPT only exploits conventional reflecting-only RISs \cite{Energy Efficiency,Enhanced secure,Intelligent Reflect,Joint Active} with only half-space coverage. Our motivation is to
fill this knowledge-gap and
to consider secure STAR-RIS aided SWIPT networks, which to date have not been studied in the literature. The main contributions of our work are summarized as follows:

\begin{itemize}
	\item We conceive a secure MISO STAR-RIS-aided SWIPT solution for the RSMA downlink, where a multi-antenna BS serves the single-antenna IRs and UERs in the transmission and reflection spaces of the STAR-RIS, respectively. The UERs are also potential \textit{Eves} that wiretap the channel.
	Moreover, we assume that the CSI of the UERs is not known at the transmitter, i.e., only the imperfectly estimated channel between the RIS and UERs is available at the BS. 
	
	\item Although most work on STAR-RIS is based on OMA/NOMA comparisons, we demonstrate that RSMA can achieve better performance than OMA/NOMA, and can provide AN for UERs, which facilitates their secure communication.
	
	\item We show how to jointly optimize the active and passive beamforming vectors for maximizing the worst-case sum secrecy rate of the IRs, under realistic constraints on the total transmit power, the minimum sum energy of the UERs, as well as the  RIS transmission and reflection coefficients. The resultant problem is non-convex with intricately coupled variables. To tackle this challenge, a suboptimal two-step iterative algorithm based on the sequential parametric convex approximation (SPCA) method of \cite{Hamed-majid} is proposed for alternately solving for the TPCs as well as  the transmission and reflection coefficients. Furthermore, an initialization algorithm is provided to 
	search for a feasible initial point of the original problem and avoid potential failures.
	
	\item Rich simulations are provided for characterizing
	 the proposed strategy and to demonstrate how increasing the number of  STAR-RIS elements leads to dramatic improvement in the system performance. Our simulation results show that
  it is the combination of
	  the RSMA-based algorithm and and the STAR-RIS that improves the efficiency. Specifically, by increasing the number of RIS elements the performance of the network is drastically enhanced.  Furthermore, our method surpasses both OMA and NOMA in performance.
\end{itemize}
To summarize the aforementioned discussion, Table \ref{tab:comparison} presents a concise comparison of our approach with the state-of-the-art, underscoring the novel aspects of our methodology in the context of existing research.
\begin{table*}[h!]
\centering
\caption{Comparative Analysis of Related Works}
\label{tab:comparison}
\begin{tabular}{|c|c|c|c|c|c|c|c|c|c|c|c|}
\hline
\diagbox[dir=SE]{\textbf{\raisebox{5pt}{Keywords}}}{\textbf{\raisebox{-5pt}{Ref.}}} & \textbf{Our Work} & \textbf{[21]} & \textbf{[22],[24]} &\textbf{[23]} &\textbf{[25]} & \textbf{[26]} & \textbf{[27]} &  \textbf{[28]} &\textbf{[29],[30]} & \textbf{[31]} & \textbf{[32]} \\ \hline
PLS                        & $\checkmark$ & & &  $\checkmark$            &   &               &  & $\checkmark$&$\checkmark$  & $\checkmark$ &               \\ \hline
RIS                        &  & &             &  &               &               &               & &$\checkmark$  & $\checkmark$ & $\checkmark$ \\ \hline
STAR-RIS                   & $\checkmark$ & $\checkmark$& $\checkmark$  &$\checkmark$ & $\checkmark$  & $\checkmark$  & $\checkmark$ & $\checkmark$ &             &               &               \\ \hline
SWIPT                      & $\checkmark$ &  &             & &               &               &               & &$\checkmark$  & $\checkmark$ & $\checkmark$ \\ \hline
RSMA                       & $\checkmark$ & &             & &        &       &               &               &               &               &               \\ \hline
NOMA                       &               & $\checkmark$&$\checkmark$ & $\checkmark$&   & $\checkmark$  & &           &    &               &               \\ \hline
OMA                        &             & &$\checkmark$ & &               &      &         &               &               &               &               \\ \hline
WCSSR                      & $\checkmark$ & & &             &               &               &  & $\checkmark$  &  &         &      \\ \hline
MISO                    & $\checkmark$ &$\checkmark$  & &             &     $\checkmark$          &               &               & $\checkmark$  & $\checkmark$ &           &  $\checkmark$  \\ \hline
I-CSI                      & $\checkmark$ &  &             &  $\checkmark$ &            &               &               & &  &  &               \\ \hline
\end{tabular}
\end{table*}

\subsection{Organization and Notation}
The rest of this paper is organized as follows.  Section \ref{Sys_Model}
introduces the system model and preliminaries, while Section \ref{Proposed method} presents the proposed Worst-Case Sum Secrecy Rate (WCSSR) maximization method. Section~\ref{Overal method} details the steps of the	
proposed algorithm as well as its complexity. Then, several experiments are performed for validating the efficiency of the proposed approach in Section \ref{Simulation Results}. Finally, Section \ref{conc} concludes the paper.

\textit{Notation}: We use bold lowercase letters for vectors and bold uppercase letters for matrices. The notation $(\cdot)^T$ and $(\cdot)^H$  denote the transpose operator and the conjugate transpose operator, respectively. $\mathfrak{Re}$ and $\mathfrak{Im}$ represent the real and the imaginary parts of a complex variable, respectively.  $\triangleq$ denotes a definition. $\mathbb{R}^{N \times 1}$ and $\mathbb{C}^{N \times 1}$ denote the sets of $N$-dimensional real and complex vectors, respectively.  $\mathbb{C}^{M \times N}$ stands for the set of $M \times N$ complex matrices.  The matrix $\mathbf{I}_N$ denotes the $N \times N$ identity matrix. 
$[\mathbf{x}]_m$ is the
 $m$-th element of a
vector $\mathbf{x}$ 
and $[\mathbf{X}]_{m,n}$ is the $(m,n)$-th element of a matrix $\mathbf{X}$. The operators $\mathrm{Tr}(\cdot)$ and $\mathrm{rank}(\cdot)$ denote the trace and rank
of a matrix, while $\mathrm{diag}\{\cdot\}$ constructs a diagonal matrix from its vector argument.

Table \ref{Table 1} presents the main parameters and variables associated with this study in order to enhance the readability of the paper. 
\begin{table}
	\small
	\renewcommand{\arraystretch}{1.3}
	\caption{Key notations used in this paper.}
	\centering
	\label{Table 1}
	\resizebox{\columnwidth}{!}{
		\begin{tabular}{|c|p{60mm}|}
			\hline
			$\mathbf{Notation}$  &  $\mathbf{Definition}$ \\
			\hline
			$\mathcal{K} /K /k$  &  Set/number/index of IRs in the transmission space \\		\hline
			$\mathcal{J} /J /j$ &  Set/number/index of UERs in the reflection space  \\		\hline
			$\mathbb{M} /M /m$ &  Set/number/index of RIS elements  \\ 	\hline
			$\beta_m^t / \beta_m^r, \theta_m^t/ \theta_m^r$ &  The amplitude and phase shift of transmission /reflection
			coefficients  \\ 		\hline
			$\mathbf{u}_t / \mathbf{u}_r, \mathbf{\Theta}_t /\mathbf{\Theta}_r $ & The transmission/reflection beamforming vector and
			diagonal matrix \\ 		\hline
			$N_T$ &  The number of transmit antennas at the BS \\ 		\hline
			$\mathbf{h}_{t,k}/\mathbf{h}_{r,j}$ &     The combined channel of the BS to the $k$/$j$th user in the transmission/reflection space	 \\ 		\hline
			$\mathbf{H}$ &  channel matrix from the BS to STAR-RIS \\ 		\hline
			$\mathbf{g}_{t,k}/\mathbf{g}_{r,j}$ &     The channel vector of the RIS to the $k$/$j$th user \\		\hline
			$\mathbf{p}_c/\mathbf{p}_k/\mathbf{f}_j$ &  The precoders for the common/private/energy signals of IRs/UERs  \\ 		\hline
			$P_t, r_c/\gamma_k$  & The total available transmit power and minmum QoS requirement of common/private stream of the $k$th user\\		\hline
			$\gamma_{c,k}^\mathrm{IR}/\gamma_{k}^\mathrm{IR}, \gamma_{c,j}^\mathrm{UER}/\gamma_{k,j}^\mathrm{UER}$ & The SINR of  common/private stream of the $k$th IR at $k/j$th IR/UER \\	\hline
			$R_{c,k}/R_{k}, R_{c,j}^\mathrm{UER}/R_{k,j}^\mathrm{UER}$ & The achievable rates of common/private stream of $k$th IR at UER-$j$
			\\		\hline
			 $R_{sec,k}^{tot}$ & The achievable secrecy rate between the BS-RIS and each legitimate user IR-k
			\\		\hline
		\end{tabular}}
	\end{table}
	
\section{System Model and Preliminaries}\label{Sys_Model}

\subsection{Assumptions and Justification }
Before presenting the system model, we briefly describe our main assumptions and justify them. 
\begin{itemize}
	
	\item We assume that users on one side of the STAR-RIS are unable to wiretap the channels of the other side. A similar assumption is made for secure transmition in an uplink NOMA network  in \cite{STAR-RIS-NOMA}.
	
	\item We assume that the users within the reflection space rely on energy harvesting while  users within the transmission space are information receivers. However, their roles may also be readily reversed. It is trivial to show that a similar formulation may be obtained in the new model, hence we will not discuss this case  further. 
	
	\item We assume that CSI is available for the IRs, since the IRs typically have two-way interactions with the BS. Hence the CSI for the IRs will be much more accurate than for the UERs, since it can be readily estimated by the BS using uplink signals gleaned from the IRs. This is a reasonable model that takes the most dominant source of CSI error into account. We also note that this assumption is very commonly used in the literature \cite{Energy Efficiency,Enhanced secure,Intelligent Reflect,Joint Active}.

    \item
    We assume the BS desires to send energy to potentially untrusted users, consistent with networks that provide multiple tiers of applications to their users. The network may offer different services to different users, depending for example on how much they have paid for access, or on their level of trustworthiness. For example, the network may be willing to provide energy for charging a device (e.g., in return for payment) without allowing it access to communication resources. Users who pay for communication access may desire that their communications enjoy some level of protection from unauthorized access by users who have not yet achieved a certain level of trust from the network. 
\end{itemize}
\subsection{STAR-RIS Signal Model}
This section describes the signal propagation model for a STAR-RIS. Fig. \ref{fig1} shows that a signal incident on a STAR-RIS element is divided into a refracted and a reflected signal in the transmission and reflection spaces, respectively.
Assume that $s_m$ is the signal incident upon
the $m$-th RIS element, where $m \in \mathbb{M} \triangleq \{1,2,\cdots,M\}$, and $M$ is the number of RIS elements. The signals that are transmitted and reflected by the
$m$-th element of the RIS are respectively modelled as \cite{Joint design}
\begin{eqnarray}
t_m & = & \left(\sqrt{\beta_m^t}e^{j\theta_m^t}\right)s_m \\ r_m & = & \left(\sqrt{\beta_m^r}e^{j\theta_m^r}\right)s_m,
\end{eqnarray}
where $\left\{\sqrt{\beta_m^t}, \sqrt{\beta_m^r} \in [ 0,1]\right\}$ and  $\left\{\theta_m^t, \theta_m^r \in [ 0,2\pi)\right\}$ denote the amplitude and phase shift of the
transmission and reflection coefficients of the $m$-th  element, respectively.

\begin{figure} 
	\centering
	\includegraphics[width=1\linewidth]{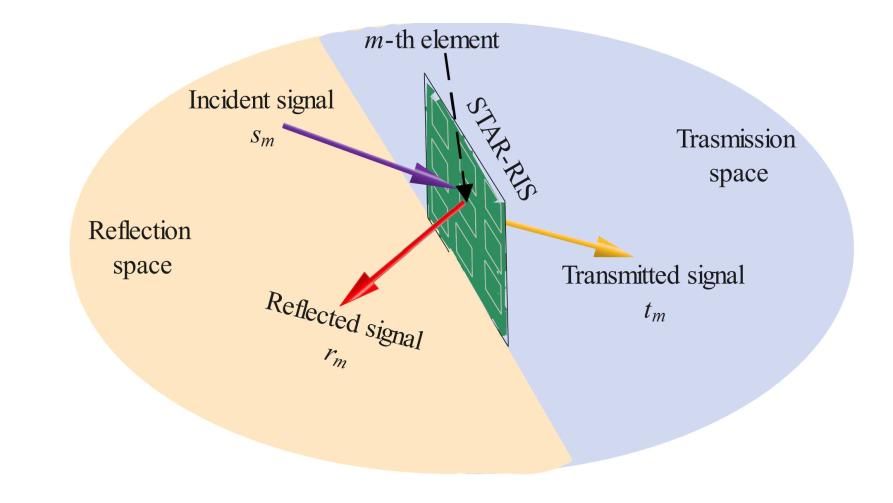}
	\caption{Illustration of signal propagation of STAR-RISs \cite{Joint design}}
	\label{fig1}		
\end{figure}

We assume that the phase shifts of the transmission and reflection elements, i.e., $\{\theta_m^t,\theta_m^r\}$ do not depend on each other. However, 
because of the energy conservation law, the incident signal's energy has to be equal to the sum of the energies of the transmitted and reflected signals i.e., $|s_m|^2 = |t_m|^2+ |r_m|^2 , \forall m \in \mathbb{M}$ and consequently we have $\beta_m^t+\beta_m^r=1 , \forall m \in \mathbb{M}$. Thus, the amplitudes of the transmission and reflection coefficients $\{\sqrt{\beta_m^t},\sqrt{\beta_m^r}\}$ are coupled.
We also assume that each STAR-RIS element operates in the energy splitting (ES) mode, i.e., simultaneous transmission and reflection mode (T $\&$ R mode) which is more general than   
the full transmission mode (T mode) or the full reflection mode (R mode) \cite{STAR1,STAR2}. The ES mode increases the number of degrees of freedom (DoFS) provided by the RIS for optimizing the network, but it also correspondingly increases the communication overhead between the BS and RIS \cite{STAR aided}.

To simplify the notation, we represent  the transmission/reflection beamforming vector by $\mathbf{u}_p = [\sqrt{\beta_1^p} e^{j\theta_1^p}, \sqrt{\beta_2^p} e^{j\theta_2^p}, \cdots, \sqrt{\beta_M^p} e^{j\theta_M^p}]$, where $p = t$ and $p = r$ distinguish either the beamformer is for transmission or reflection, respectively, and the corresponding diagonal transmission/reflection   matrix of the STAR-RIS is given by  $\mathbf{\Theta}_p=\mathrm{diag}(\mathbf{u}_p)$. Equation~\eqref{Eq2} summarizes the set of constraints for the transmission and reflection coefficients.
\begin{figure*}
	\begin{align}\label{Eq2}
	\mathbb{R}_{\beta,\theta} =\left\{ \beta_m^t, \beta_m^r, \theta_m^t, \theta_m^r \ \big| \beta_m^t, \beta_m^r \in [0,1]; \beta_m^t+ \beta_m^r  =1; \theta_m^t, \theta_m^r \in [0, 2\pi)     \right \}.
	\end{align}
	\hrulefill
\end{figure*}

\subsection{System Model}
We consider downlink transmission in a STAR-RIS assisted SWIPT network, where the direct links between the BS and single-antenna users are blocked and the BS communicates with the users only with the aid of the STAR-RIS. We denote the number of transmit antennas at the BS and the number of STAR-RIS elements by $N_T$ and $M$ respectively. 
The users are divided into two types: The IRs in the transmission space of the RIS and the UERs in the reflection area, as shown in Fig. \ref{fig2}. 
The $K$ single-antenna IRs are indexed by $\mathcal{K} = \{1,\dots, K\}$ and the $J$ single-antenna UERs are indexed by $\mathcal{J} = \{1, \dots, J\}$. The IRs and UERs are respectively assumed to possess Information Decoding (ID) and Energy Harvesting (EH) capabilities. 

The combined channel of the BS-RIS-user link for user $n$ for either the IRs or UERs is denoted by $\mathbf{h}_{c,n}=\mathbf{g}_{c,n}^H \boldsymbol{\Theta}_n \mathbf{H}$, where $\mathbf{H} \in \mathbb{C}^{M\times N_T}$ represents the channel matrix of the link  from the BS to the RIS, and $\mathbf{g}_{c,n} \in \mathbb{C}^{M\times 1}$  is the channel vector of the link from the RIS to user $n$ which is defined as
\begin{align}
\mathbf{g}_{c,n}=
\begin{cases}
\mathbf{g}_{t,k}, & \text{if user $n$ is located in the transmission space}\\
\mathbf{g}_{r,j}, & \text{if user $n$ is located in the reflection space.}
\end{cases}
\end{align}
The CSI of the IRs is assumed to be known at the transmitter. However, due to lack of frequent and steady interaction between the BS and UERs, the CSI of the UERs may become outdated during the transmission. Therefore, the BS only has access to the estimates $\hat{\mathbf{g}}_{r,j}$ of the channel between the BS and UERs. We use a bounded error model \cite{Imperfection1,Imperfection2} to account for the imprecise CSI of the UERs, as follows:
\begin{align} \label{g-hat}
\mathbf{g}_{r,j}&=\hat{\mathbf{g}}_{r,j}+\Delta \mathbf{g}_{r,j}, \nonumber\\ \Theta_g &\triangleq \{ \mathbf{\Delta g}_{r,j} \in \mathbb{C}^{M \times 1}:\Vert  \mathbf{\Delta g}_{r,j}\Vert^2 \leq \nu^2  \},
\end{align}
where $\hat{\mathbf{g}}_{r,j} \in \mathbb{C}^{M\times 1}$ is the estimate of the channel between the RIS and the UERs, which is available to the BS, and $\Delta \mathbf{g}$ represents the unknown channel uncertainty of the UERs. The variable $\nu>0$ denotes the size of the uncertainty region of the estimated UER CSI. The special case of perfect CSI for the  RIS to UER channels occurs when $\nu=0$.
The diagonal matrix $\boldsymbol{\Theta}_n $ of RIS coefficients  and the combined channel $\mathbf{h}_{c,n}$ are respectively defined as
\begin{align}
\boldsymbol{\Theta}_n=
\begin{cases}
\boldsymbol{\Theta}_t, & \text{if user $n$ is located in the IR space}\\
\boldsymbol{\Theta}_r, & \text{if user $n$ is located in the UER space},
\end{cases}
\end{align}
\begin{align}
\mathbf{h}_{c,n}=
\begin{cases}
\mathbf{h}_{t,k}, & \text{if user $n$ is located in the IR space}\\
\mathbf{h}_{r,j}, & \text{if user $n$ is located in the UER space}.
\end{cases}
\end{align}

\begin{figure} 
	\centering
	\includegraphics[width=1\linewidth]{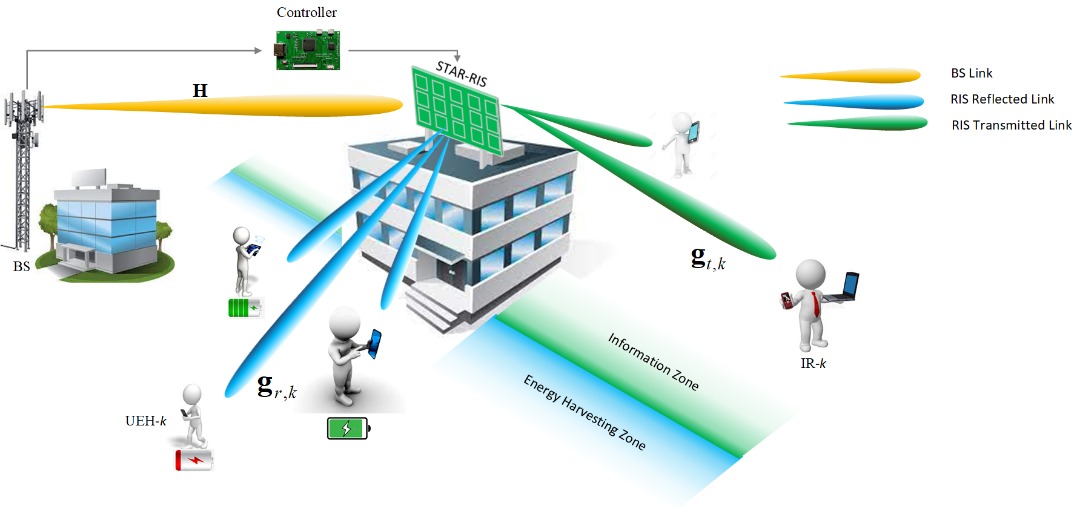}
	\caption{System model}
	\label{fig2}		
\end{figure}
As assumed in the literature on RSMA \cite{RSMA,RS1,RS2}, the BS divides the message $W_k$ of IR-$k$ into a common part $W_{c,k}$ and a private part $W_{p,k}, \forall k \in \mathcal{K}$.
The common stream $s_c^{\mathrm{ID}}$ is created by jointly
encoding the common parts of all IRs $\{ W_{c,1},\dots, W_{c,K}\}$. The private streams $\{s_1^{\mathrm{ID}},\dots,s_K^{\mathrm{ID}}\}$ are obtained by 
independently encoding the private parts of the IRs. Therefore, the aggregated transmit signal is given by
\begin{align}
\mathbf{x}=\mathbf{x}^{\mathrm{ID}}+\mathbf{x}^{\mathrm{EH}}=\mathbf{p}_c s_c^{\mathrm{ID}}+\sum_{k\in \mathcal{K}} \mathbf{p}_k s_k^{\mathrm{ID}}+ \sum_{j\in \mathcal{J}} \mathbf{f}_j s_j^{\mathrm{EH}},
\end{align}
where $\mathbf{p}_c,\mathbf{p}_k$, and $\mathbf{f}_j \in \mathbb{C}^{N_T \times 1}$ are the precoders for the common, the private and the energy signals, respectively. We denote by $\mathbf{P}=[\mathbf{p}_c,\mathbf{p}_1,\dots,\mathbf{p}_K]$ and $\mathbf{F}=[\mathbf{f}_1,\dots,\mathbf{f}_J]$  the information precoder and energy precoder matrices, respectively. The sets  $\mathbf{s}^{\mathrm{ID}}=[s^{\mathrm{ID}}_c,s^{\mathrm{ID}}_1,\dots,s^{\mathrm{ID}}_K]^T \in \mathbb{C}^{(K+1) \times 1}$ and $\mathbf{s}^{\mathrm{EH}}=[s^{\mathrm{EH}}_1,\dots,s^{\mathrm{EH}}_J]^T \in \mathbb{C}^{J\times 1}$ represent the IR and UER streams, respectively.
Under the assumption that $\mathbb{E}\{\mathbf{s}^{\mathrm{ID}} (\mathbf{s}^{\mathrm{ID}})^H \}= \mathbf{I}_K $ and
$ \mathbb{E}\{\mathbf{s}^{\mathrm{EH}} (\mathbf{s}^{\mathrm{EH}})^H \}= \mathbf{I}_J $, the transmit power constraint is $ \mathrm{Tr}(\mathbf{P}\mathbf{P}^H )+\mathrm{Tr}(\mathbf{F}\mathbf{F}^H ) \leq  P_t$, where $P_t$ is the total transmit power available.

The signals received at IR-$k$ and UER-$j$  are respectively given by
\begin{equation} 
\begin{split}
y_k^\mathrm{ID}   & = \mathbf{h}_{t,k} \mathbf{x} + n_k^\mathrm{ID}, \  \ \forall k \in \mathcal{K},\\
y_j^\mathrm{EH}  & =  \mathbf{h}_{r,j} \mathbf{x} + n_j^\mathrm{EH}, \ \ \forall j \in \mathcal{J},
\end{split}
\end{equation}
where $\mathbf{h}_{t,k},\mathbf{h}_{r,j} \in \mathbb{C}^{1\times M} $ are the corresponding combined channels from the BS-RIS to IR-$k$ and from the BS-RIS to UER-$j$, respectively. The terms $n_k^\mathrm{ID}$ and $n_j^\mathrm{EH}$ represent  additive Gaussian noise with zero mean and unit variance received at IR-$k$ and UER-$j$, respectively.

We assume that the energy precoder  $\mathbf{F}$ is perfectly known at the transmitter and IRs, and thus the IRs remove the interference of the energy signals from $y_k^\mathrm{ID}$ to decode the desired information signals. According to the standard RSMA decoding order \cite{RS1,RS2,RSMA}, before decoding the intended private stream, the common stream is decoded by each IR, while treating the private streams as interference. Hence, the Signal-to-Interference-plus-Noise Ratio (SINR) experienced upon decoding the common stream $s_c^\mathrm{ID}$ at IR-$k$ is given by
\begin{equation} \label{gam-ck}
\gamma_{c,k}^\mathrm{IR}=\dfrac{\left|\mathbf{h}_{t,k}\mathbf{p}_{c}\right|^2}{\sum_{k' \in \mathcal{K}}\left|\mathbf{h}_{t,k}\mathbf{p}_{k'}\right|^2+\sum_{j \in \mathcal{J}}\left|\mathbf{h}_{t,k}\mathbf{f}_{j}\right|^2+1}, \forall  k \in \mathcal{K}.
\end{equation}
The contribution of $s_c^\mathrm{ID}$ is removed from $y_k$ after decoding. Then, the intended private stream $s_k^\mathrm{ID}$ of IR-$k$ can be decoded by treating the interference arriving from other IRs as noise. The SINR for decoding the private stream $s_k^\mathrm{ID}$ at IR-$k$ is formulated as:
\begin{equation} \label{gam-k}
\gamma_{k}^\mathrm{IR}=\dfrac{\left|\mathbf{h}_{t,k}\mathbf{p}_{k}\right|^2}{\sum_{k' \in \mathcal{K},k' \neq k}\left|\mathbf{h}_{t,k}\mathbf{p}_{k'}\right|^2+\sum_{j \in \mathcal{J}}\left|\mathbf{h}_{t,k}\mathbf{f}_{j}\right|^2+1}, \forall  k \in \mathcal{K}.
\end{equation}
The achievable information rates of $s_c^\mathrm{ID}$ and $s_k^\mathrm{ID}$ at IR-$k$ can be respectively written as:
\begin{eqnarray}
R_{c,k} &=& \mathrm{log}_2 \left( 1+ \gamma_{c,k}^\mathrm{IR}  \right), \ \ \forall  k \in \mathcal{K} \\
R_{k} &=& \mathrm{log}_2 \left( 1+ \gamma_{k}^\mathrm{IR}  \right), \ \ \forall  k \in \mathcal{K}.
\end{eqnarray}
Since it is required that all IRs decode the common stream, the achievable rate of $s_c^\mathrm{ID}$ must not exceed the minimum achievable rate of all IRs, i.e., $R_c=\mathrm{min} \{ R_{c,1},\dots, R_{c,K}\}$. Assume that the common rate $R_c$ is shared by all $K$ IRs, and let $C_k$ denote the specific portion of $R_c$ due to transmitting $W_{c,k}$, so that we have $C_k=\alpha_k R_c $ where $\sum_{k=1}^{K} \alpha_k =1$. Therefore, the overall achievable rate of IR-$k$ is given by $R_{k,tot}=C_k+R_k$, which includes the contributions of the common and private rates $W_{c,k}$ and $W_{p,k}$, respectively.

The total harvested energy at UER-$j$ is the sum of the energy carried by all information and energy precoders, which can be written as \cite{harvested energy}:
\begin{equation} \label{Q-j def} 
Q_j = \zeta \left(|\mathbf{h}_{r,j}\mathbf{p}_{c}|^2  
+ \sum_{k' \in \mathcal{K}}\left|\mathbf{h}_{r,j}\mathbf{p}_{k}\right|^2+\sum_{j \in \mathcal{J}}\left|\mathbf{h}_{r,j}\mathbf{f}_{j}\right|^2  \right), \forall  j \in \mathcal{J},
\end{equation}
where $0 \leq \zeta \leq 1$ is the energy harvesting efficiency. We assume having $ \zeta = 1$ in the rest of the paper.

It is assumed that the UERs are potential eavesdroppers that can wiretap the IR channels and that know the precoder matrices $\mathbf{P}$ and $\mathbf{F}$. Thus, the SINR for decoding the common stream $s_c^\mathrm{ID}$ and the private stream $s_k^\mathrm{ID}$ at UER-$j$ are:
\begin{align} 
\gamma_{c,j}^\mathrm{UER}=& \dfrac{\big|\mathbf{h}_{r,j}\mathbf{p}_{c}\big|^2}{\displaystyle \sum_{k' \in \mathcal{K}}\big|\mathbf{h}_{r,j}\mathbf{p}_{k'}\big|^2+\displaystyle \sum_{j' \in \mathcal{J}}\big|\mathbf{h}_{r,j}\mathbf{f}_{j'}\big|^2+1}, \forall  j \in \mathcal{J} \\
\gamma_{k,j}^\mathrm{UER}=& \dfrac{\big|\mathbf{h}_{r,j}\mathbf{p}_{k}\big|^2}{\big|\mathbf{h}_{r,j}\mathbf{p}_{c}\big|^2+\displaystyle\sum_{k' \in \mathcal{K}, k' \neq k}\big|\mathbf{h}_{r,j}\mathbf{p}_{k'}\big|^2+\displaystyle\sum_{j \in \mathcal{J}}\big|\mathbf{h}_{r,j}\mathbf{f}_{j}\big|^2+1}, \nonumber\\ & \forall  k \in \mathcal{K}, j \in \mathcal{J}.
\end{align}
The corresponding  rates for  $s_c^\mathrm{ID}$ and $s_k^\mathrm{ID}$ achievable at UER-$j$ are $R_{c,j}^\mathrm{UER}=\mathrm{log}_2 \left( 1+ \gamma_{c,j}^\mathrm{UER}  \right) $, and  $R_{k,j}^\mathrm{UER}=\mathrm{log}_2 \left( 1+ \gamma_{k,j}^\mathrm{UER}  \right) $. For the UER, we aim for designing the transmission power to ensure that the information messages are not decodable at the UER. In order to achieve this goal, the condition $\displaystyle \max_{j \in \mathcal{J}} \left\{  \displaystyle \max_{ \mathbf{\Delta g}_{r,j}} \left\{ R_{c,j}^\mathrm{UER} \right\}   \right\} <R_c$ should be satisfied for the common signal. 

We emphasize that we assume a worst-case scenario in terms of security, where the UER knows both its own CSI and that of all precoders. Thus, the achievable secrecy rate between the BS-RIS and each legitimate user IR-k is given by:
\begin{align} 
R_{sec,k}^{tot} \triangleq & \ R_{sec,c} + R_{sec,k},
 \ \forall k \in \mathcal{K} \nonumber \\
R_{sec,c} \triangleq & \ \alpha_k \left[ R_c - \max_{ j \in \mathcal{J}}  \left \{  R_{c,j}^\mathrm{UER} \right \} \right]^{+},\nonumber \\
R_{sec,k} \triangleq & \left[ R_k - \max_{ j \in \mathcal{J}}  \left \{  R_{k,j}^\mathrm{UER} \right \} \right]^{+},\ \forall k \in \mathcal{K}
\end{align}
where we define the operation $[x]^+ \triangleq \max(0,x)$. Since the CSI between the BS and UERs is imperfectly known, the actual secrecy rate between the BS-RIS and each legitimate user IR-$k$ for the worst case UER CSI estimation error is defined as:
\begin{align} \label{Rsec-def}
\hat{R}_{sec,k}^{tot} \triangleq & \ \hat{R}_{sec,c} + \hat{R}_{sec,k},
\ \forall k \in \mathcal{K} \nonumber \\
\hat{R}_{sec,c} \triangleq & \ \alpha_k \left[ R_c - \max_{ j \in \mathcal{J}}\left  \{ \max_{ \mathbf{\Delta g}_{r,j}} \{  R_{c,j}^\mathrm{UER} \}\right \} \right]^{+},\nonumber \\
\hat{R}_{sec,k} \triangleq & \left[ R_k - \max_{ j \in \mathcal{J}}\left  \{ \max_{ \mathbf{\Delta g}_{r,j}} \{  R_{k,j}^\mathrm{UER} \}\right \} \right]^{+}.
\end{align}

\section{Proposed WCSSR Maximization}\label{Proposed method}
\subsection{Problem Formulation}
Before proceeding further, we have provided a diagram,
depicted in Fig. \ref{fig11}, to show the flow of the analysis described
in the sequel.

\begin{figure} 
	\centering
	\includegraphics[width=1\linewidth]{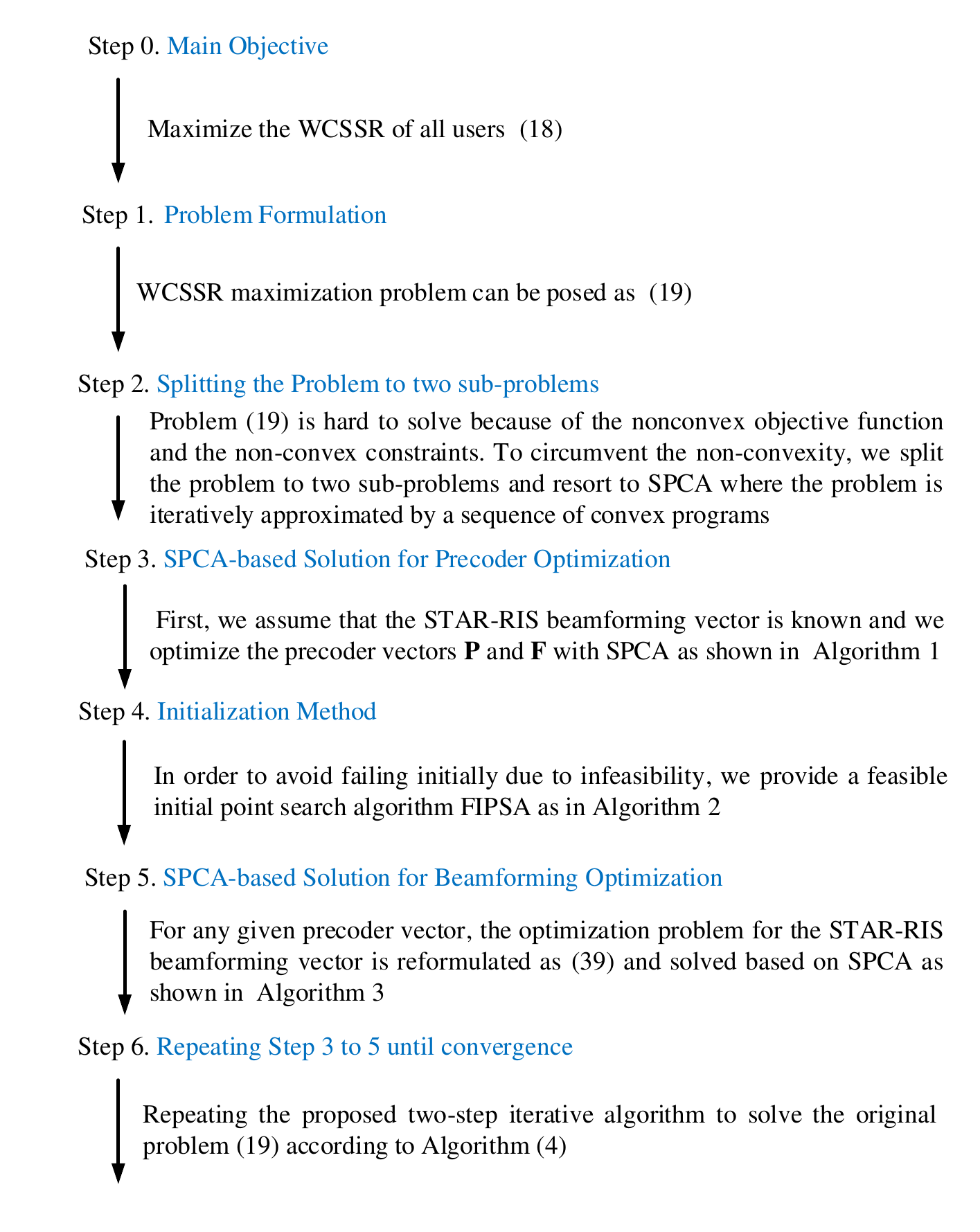}
	\caption{Analyzing the problem: A step-by-step overview}
	\label{fig11}		
\end{figure}

Our objective is to maximize the WCSSR of all users with proportional rate constraints. Mathematically, the WCSSR maximization problem can be formulated as follows:
\begin{subequations} \label{P1}
	\allowdisplaybreaks
\begin{align} 
&\max_{\mathbf{P},\mathbf{F},\mathbf{u}_p,\boldsymbol{\alpha}}\left(\min_k \left\{\hat{R}_{sec,k}^{tot}\right\}\right), \label{P1-a}\\
\mathrm{s.t.}\ & 
0 \leq \alpha_k \leq 1, \ \forall k \in \mathcal{K},\label{P1-b}\\
&\sum_{k \in \mathcal{K}} \alpha_k =1,\label{P1-c}\\
& R_c \leq R_{c,k} , \forall k \in \mathcal{K},\label{P1-d}\\
&\alpha_k R_c \geq r_c, \ \forall k \in \mathcal{K} \label{P1-e}\\
&\sum_{j \in \mathcal{J}} \min_{\mathbf{\Delta g}_{r,j}} \{Q_j\} \geq E^{th},\label{P1-f}\\
&\mathrm{Tr} ( \mathbf{P}\mathbf{P}^H ) +\mathrm{Tr}(\mathbf{F}\mathbf{F}^H)\leq P_t, \label{P1-g}\\
& \beta_m^p, \theta_m^p \in  \mathbb{R}_{\beta,\theta}, \forall m \in \mathbb{M}, \forall p \in \{t,r\}, \label{P1-h}
\end{align}
\end{subequations}
where $r_c$ is a predefined threshold, $\boldsymbol{\alpha} \triangleq [\alpha_1,\alpha_2, \cdots, \alpha_K]^T$, and $E^{th}$ is a threshold representing the minimum harvested energy required by the UERs. Note that in constraint \eqref{P1-f},  we use the worst-case scenario for the sum of the energy harvested by the UERs due to the imperfect CSI of the link between the RIS and the UERs. Problem \eqref{P1} is difficult to solve because of the nonconvex nature of the objective function and the non-convex constraints. Thus, finding the global optimum is generally intractable. To circumvent the non-convexity, we resort to SPCA, where the problem is iteratively approximated by a sequence of convex programs. At each iteration, the non-convex constraints are replaced by convex surrogates that serve as approximations.

\subsection{SPCA-based Solution for Precoder Optimization}
We assume here that the STAR-RIS beamforming vector $\mathbf{u}_p$ is known and we want to optimize the precoders $\mathbf{P},\mathbf{F}$.
To use SPCA, we first mitigate the non-convexity using some variable transformations and linearization. Then, a well-suited convex subset is constructed via SPCA that approximates the original non-convex solution set. We adopt an iterative solution to cope with the nonconvexity by approximating the non-convex factor at each iteration by its first order Taylor expansion. In what follows, we independently examine each non-convex term and derive its convex counterpart.

\subsubsection{Non-Convex Objective Function \eqref{P1-a}}
We first eliminate the inner minimization of the max-min problem~\eqref{P1} that finds the worst-case CSI model in the bounded set. To do so, we introduce the auxiliary variable $r_{sec}$ and reformulate the objective together with the constraints of \eqref{P1} as follows:
\begin{subequations} \label{P2}
		\allowdisplaybreaks
	\begin{align} 
	&\max_{\mathbf{P},\mathbf{F},\boldsymbol{\alpha}} \ r_{sec} \label{P2-a}\\
	\mathrm{s.t.}\ & 
	\alpha_k \left[ R_c - \max_{ j \in \mathcal{J}}\left  \{ \max_{ \mathbf{\Delta g}_{r,j}} \{  R_{c,j}^\mathrm{UER} \}\right \} \right]^{+} + \nonumber \\ &\left[ R_k - \max_{ j \in \mathcal{J}}\left  \{ \max_{ \mathbf{\Delta g}_{r,j}} \{  R_{k,j}^\mathrm{UER} \}\right \} \right]^{+} \geq r_{sec}, \ \forall k \in \mathcal{K}  \label{P2-b}\\
	& \eqref{P1-b}-\eqref{P1-g}.
	\end{align}
\end{subequations}
It can be seen that $r_{sec}$ plays
the role of a lower bound for $\min_k \left\{\hat{R}_{sec,k}^{tot}\right\}$, and its
maximization will increase the left-side of the constraint
\eqref{P2-b}, so that the constraint is active at the optimum.
We note that to arrive at \eqref{P2}, $\hat{R}_{sec,c}$ and $\hat{R}_{sec,k}$ in \eqref{P2-b} were substituted by their definitions from \eqref{Rsec-def}.

Due to the operator $[\cdot]^+$, the constraint \eqref{P2-b} is nonconvex. We will relax the problem by introducing two artificial constraints in \eqref{P3-e} and \eqref{P3-f} to replace this operator. To further facilitate convexifying  \eqref{P2}, we define two sets of new auxiliary variables $\boldsymbol{\alpha}_{c} \triangleq [\alpha_{c,1},\alpha_{c,2},\cdots,\alpha_{c,J}]^T$ and
$\boldsymbol{\alpha}_{p} \triangleq [\boldsymbol{\alpha}_{p,1},\boldsymbol{\alpha}_{p,2},\cdots,\boldsymbol{\alpha}_{p,K}]^T$, where $\boldsymbol{\alpha}_{p,k} \triangleq [\alpha_{p,k,1},\alpha_{p,k,2},\cdots,\alpha_{p,k,J}]^T, \ \forall k \in \mathcal{K}$. This
allows us to transform \eqref{P2}  into:
\begin{subequations} \label{P3}
		\allowdisplaybreaks
	\begin{align} 
	&\max_{\mathbf{P},\mathbf{F},\boldsymbol{\alpha},\boldsymbol{\alpha}_{c},\boldsymbol{\alpha}_{p}} \ r_{sec} \label{P3-a}\\
	\mathrm{s.t.}\ & 
	\alpha_k(R_c-\alpha_{c,j})+\gamma_k-\alpha_{p,k,j} \geq  r_{sec},  \ \forall k \in \mathcal{K}, \label{P3-b}\\	
	&  \max_{ \mathbf{\Delta g}_{r,j}} \left\{R_{c,j}^{\mathrm{UER}}\right\} \leq \alpha_{c,j}, \ \forall j \in \mathcal{J}, \label{P3-c}\\
	&  \max_{ \mathbf{\Delta g}_{r,j}} \left\{R_{k,j}^{\mathrm{UER}}\right\} \leq \alpha_{p,k,j}, \ \forall j \in \mathcal{J}, \label{P3-d}\\
	& R_c \geq \alpha_{c,j}, \ \forall j \in \mathcal{J},\label{P3-e} \\
	& \gamma_k  \geq \alpha_{p,k,j}, \ \forall j \in \mathcal{J}, k \in \mathcal{K}, \label{P3-f}\\
	& \gamma_k  \leq R_k, \ \forall k \in \mathcal{K}, \label{P3-g}\\
	& \eqref{P1-b}-\eqref{P1-g}.
	\end{align}
\end{subequations}

Based on the above discussion, $\gamma_k$ represents a lower-bound for $R_k$, while $\alpha_{c,j}$ and $\alpha_{p,k,j}$ serve as upper-bounds for  $\max_{ j \in \mathcal{J}}\left  \{ \max_{ \mathbf{\Delta g}_{r,j}} \{  R_{c,j}^\mathrm{UER} \}\right \}$ and $\max_{ j \in \mathcal{J}}\left  \{ \max_{\mathbf{\Delta g}_{r,j}} \{  R_{k,j}^\mathrm{UER} \}\right \}$, respectively. Increasing the lower-bound and simultaneously reducing the upper-bounds boosts the left-side of
the constraints, which makes the constraints
\eqref{P3-b}-\eqref{P3-g} active at the optimum.

Despite this linearization, it can be seen by invoking
the definitions of $R_{c,j}^\mathrm{UER}, R_{k,j}^\mathrm{UER}$ and $R_k$ that the constraints \eqref{P3-b}-\eqref{P3-d}, and \eqref{P3-g} are still non-convex. To handle the non-convexity of these constraints, we construct a suitable inner convex subset for approximating the nonconvex feasible solution set. Along this line, we first approximate \eqref{P3-b} by its first-order Taylor expansion to obtain
\begin{align} 
\Theta^{[i]}(\alpha_k,R_c)-\bar{\Theta}^{[i]}(\alpha_k,\alpha_{c,j}) +\gamma_k-\alpha_{p,k,j} \geq  r_{sec},	
\end{align}	
where we have defined
\begin{align}
&\Theta^{[i]}(x,y) \triangleq \\\nonumber
& \frac{1}{2} (x^{[i]}+y^{[i]})(x+y)-\frac{1}{4}(x^{[i]}+y^{[i]})^2-\frac{1}{4}(x-y)^2,
\end{align} 
and
\begin{align}
&\bar{\Theta}^{[i]}(x,y) \triangleq \\\nonumber
& \frac{1}{4} (x+y)^2 + \frac{1}{4} (x^{[i]}-y^{[i]})^2 - \frac{1}{2}(x^{[i]}-y^{[i]})(x-y)
\end{align}
for the linear approximation of the terms, which involve the product of two variables.

To handle the non-convexity of \eqref{P3-c},\eqref{P3-d}, and \eqref{P3-g} we build a suitable inner convex subset to approximate the nonconvex feasible solution set. In particular, we
first define a set of new auxiliary variables  $\boldsymbol{\pi}_{c,j,j',k} \triangleq [\rho_{c,j},\rho_{k,j},\rho_k,a_{j,k}, b_{j,j'},x_{c,j},x_{k,j},v_j], \ \forall k \in \mathcal{K}, \{j,j'\} \in \mathcal{J} $ and
exploit the following Propositions.

\begin{Proposition}
An affine approximation of constraint \eqref{P3-c}, $\forall j \in \mathcal{J}$ is given by:
\begin{align}\label{lem1-formul}
\begin{cases}
(\RN{1}): 1 + \rho_{c,j} - \Gamma^{[i]}(\alpha_{c,j}) \leq 0,
\\
(\RN{2}): \dfrac{x_{c,j}^2}{\sum_{k' \in \mathcal{K}} a_{j,k'}+\sum_{j' \in \mathcal{J}} b_{j,j'}+1} \leq \rho_{c,j},  
\\ 
(\RN{3}):\big |\hat{\mathbf{g}}_{r,j}^H \mathbf{p}'_c \big| + \nu \left \Vert \mathbf{p}'_c \right \Vert_2 \leq x_{c,j},
\\ 
(\RN{4}): \mathrm{Tr} \left[(\hat{\mathbf{G}}_{r,j}-\mu_j \mathbf{I})\mathbf{P}'_{k'}\right]  \geq a_{j,k'},
\\ 
(\RN{5}): \mathrm{Tr} \left[(\hat{\mathbf{G}}_{r,j}-\mu_j \mathbf{I}) \mathbf{F}'_{j'} \right] \geq b_{j,j'},
\end{cases}
\end{align}
where
\begin{align}
\Gamma^{[i]}(x)  & \triangleq2^{x^{[i]}} [1+\ln(2)(x-x^{[i]})], \\
\hat{\mathbf{G}}_{r,j} &\triangleq \hat{\mathbf{g}}_{r,j} \hat{\mathbf{g}}_{r,j}^H, \\
\mathbf{P}'_{k'}  &\triangleq\mathbf{p}'_{k'} (\mathbf{p}'_{k'})^{H}, \\
\mathbf{p}'_{n} &\triangleq \boldsymbol{\Theta}_r \mathbf{H} \mathbf{p}_n, \ \mathrm{for} \ n \in \{c,k\},\\
\mathbf{F}'_{j'}  &\triangleq \mathbf{f}'_{j'} (\mathbf{f}'_{j'})^{H}, \\ 
\mu_j & \triangleq  \nu^2 + 2\nu \Vert \hat{\mathbf{g}}_{r,j} \Vert_2.
\end{align}
\end{Proposition}

\begin{proof}
See Appendix \ref{app1}.
\end{proof}

\begin{Proposition}
An affine approximation of constraint \eqref{P3-d}, $ \forall j \in \mathcal{J}, k \in \mathcal{K}$ is given by:
\begin{align}
\begin{cases}
(\RN{1}): 1 + \rho_{k,j} - \Gamma^{[i]}(\alpha_{p,k,j}) \leq 0,
\\
(\RN{2}): \dfrac{x_{k,j}^2}{v_j+\sum_{k' \in \mathcal{K}, k' \neq k} a_{j,k'}+\sum_{j' \in \mathcal{J}} b_{j,j'}+1} \leq \rho_{k,j}, 
\\ 
(\RN{3}):\left |\hat{\mathbf{g}}_{r,j}^H \mathbf{p}'_k \right| + \nu \left \Vert \mathbf{p}'_k \right \Vert_2 \leq x_{k,j},
\\ 
(\RN{4}): \mathrm{Tr} \left[(\hat{\mathbf{G}}_{r,j}-\mu_j \mathbf{I}) \mathbf{P}'_{c} \right]\geq v_j,
\end{cases}
\end{align}
where $\mathbf{P}'_{c} \triangleq \mathbf{p}'_{c} (\mathbf{p}'_{c})^{H}$.
\end{Proposition}

\begin{proof}
This Proposition can be proved following the same approach as that presented in	Appendix \ref{app1}.
\end{proof}

\begin{Proposition}
	An affine approximation of constraint \eqref{P3-g}, $ \forall k \in \mathcal{K}$ is given by:
	\begin{align}
	\begin{cases}
	(\RN{1}): 1 + \rho_{k} - 2^ {\gamma_k}  \geq 0,
	\\
	(\RN{2}): \sum_{k' \in \mathcal{K},k' \neq k}\left|\mathbf{h}_{t,k}\mathbf{p}_{k'}\right|^2+\sum_{j \in \mathcal{J}}\left|\mathbf{h}_{t,k}\mathbf{f}_{j}\right|^2-\\
	\Psi^{[i]}(\mathbf{p}_{k},\rho_{k};\mathbf{h}_{t,k})+1 \leq 0,
	\end{cases}
	\end{align}
where $\Psi^{[i]}(\mathbf{u},x;\mathbf{h}) \triangleq 
\frac{2 \mathfrak{Re} \left\{ ( \mathbf{u}^{[i]})^H \mathbf{h} \mathbf{h}^H \mathbf{u} \right\} }{x^{[i]}} - \frac{\left | \mathbf{h}^H \mathbf{u}^{[i]} \right |^2 x}{(x^{[i]})^2}$. 
\end{Proposition}

\begin{proof}
See Appendix \ref{app2}.
\end{proof}

\subsubsection{Non-convex constraint \eqref{P1-d}}
To handle the nonconvexity of \eqref{P1-d}, we first introduce the new auxiliary variables $\rho_{c,k} \ \forall k \in \mathcal{K}$ and resort to Proposition 4, as follows:

\begin{Proposition}
	An affine approximation of constraint \eqref{P1-d}, $ \forall k \in \mathcal{K}$ is given by:
	\begin{align}
	\begin{cases}
	(\RN{1}): 1 + \rho_{c,k} - 2^ {R_c}  \geq 0,
	\\
	(\RN{2}): \sum_{k' \in \mathcal{K}}\left|\mathbf{h}_{t,k}\mathbf{p}_{k'}\right|^2+\sum_{j \in \mathcal{J}}\left|\mathbf{h}_{t,k}\mathbf{f}_{j}\right|^2-\\
	\Psi^{[i]}(\mathbf{p}_{c},\rho_{c,k};\mathbf{h}_{t,k})+1 \leq 0.
	\end{cases}
	\end{align}
\end{Proposition}

\begin{proof}
This Proposition is proved by following the same approach as that presented in
Appendix \ref{app2}.
\end{proof}

\subsubsection{Non-convex constraint \eqref{P1-e}}
To handle the nonconvexity of \eqref{P1-e}, we utilize Proposition 5, as follows:

\begin{Proposition}
	An affine approximation of constraint \eqref{P1-e}, $ \forall k \in \mathcal{K}$ is given by:
	\begin{align} \label{eq-23}
	\Theta^{[i]}(\alpha_k,R_c) \geq r_c. 
	\end{align}
\end{Proposition}

\begin{proof}
	See Appendix \ref{app3}.
\end{proof}

\subsubsection{Non-convex constraint \eqref{P1-f}}
To handle the nonconvexity of \eqref{P1-f}, we introduce the new auxiliary variables $\boldsymbol{\lambda}_{c,j,j',k} \triangleq [\lambda_{j,c} ,\lambda_{j,k}, \xi_{j,j'}], \ \forall k \in \mathcal{K}, \{j,j'\} \in \mathcal{J}$ and use the following Proposition:

\begin{Proposition}
	An affine approximation of constraint \eqref{P1-f}, is given by:
	\begin{align}  \label{lem6-eq}
	\begin{cases}
	(\RN{1}):\displaystyle\sum_{j \in \mathcal{J}} \left(\lambda_{j,c} + \displaystyle\sum_{k \in \mathcal{K}} \lambda_{j,k} + \displaystyle\sum_{j' \in \mathcal{J}} \xi_{j,j'}
	\right) \geq E^{th},
	\\
	(\RN{2}): \mathrm{Tr} \left[(\hat{\mathbf{G}}_{r,j}-\mu_j \mathbf{I}) \mathbf{P}'_{c} \right]\geq \lambda_{j,c},  \ \forall j \in \mathcal{J}, 
	\\
	(\RN{3}): \mathrm{Tr} \left[(\hat{\mathbf{G}}_{r,j}-\mu_j \mathbf{I}) \mathbf{P}'_{k} \right]\geq \lambda_{j,k}, \ \forall j \in \mathcal{J}, k \in \mathcal{K},   
	\\
	(\RN{4}): \mathrm{Tr} \left[(\hat{\mathbf{G}}_{r,j}-\mu_j \mathbf{I}) \mathbf{F}'_{j'}\right] \geq \xi_{j,j'},
	 \ \forall j,j' \in \mathcal{J}.
	\end{cases}
	\end{align}
\end{Proposition}

\begin{proof}
	See Appendix \ref{app4}.
\end{proof}

With the above approximations, the proposed SPCA-based approach is outlined in Algorithm \ref{Alg1}, in which the following convex optimization
problem is solved:
\begin{align} \label{eq-alg1}
	&\max_{\mathbf{x}} \ r_{sec} \nonumber\\
	\mathrm{s.t.}\ & \eqref{P3-b}, \eqref{P3-e}, \eqref{P3-f},
	\eqref{P1-b}, \eqref{P1-c}, \nonumber\\ & \eqref{lem1-formul}-\eqref{lem6-eq}, \boldsymbol{\bar{\omega}}\succeq 0
\end{align}
where $\boldsymbol{\bar{\omega}} \triangleq [\gamma_k, \boldsymbol{\alpha}_{c}, \boldsymbol{\alpha}_{p}, \boldsymbol{\pi}_{c,j,j',k}, \rho_{c,k}, \boldsymbol{\lambda}_{c,j,j',k}], \ \forall k \in \mathcal{K},\ \{j,j'\} \in \mathcal{J}$, and $\mathbf{x} \triangleq [\mathbf{P},\mathbf{F},\boldsymbol{\alpha}, \boldsymbol{\bar{\omega}}]$.
If there exist feasible initial points for \eqref{eq-alg1}, the feasible set
defined by the constraints of \eqref{eq-alg1} and consequently the resultant solutions are guaranteed to lie within the original feasible
set defined by \eqref{P1}. The procedure continues
until the stopping criterion is satisfied or the affordable number of iterations is reached.

\begin{algorithm}
	\caption{SPCA-based Algorithm for Precoder Optimization}
	\label{Alg1}
	\begin{algorithmic}[1]
		\State \textbf{Input:}
		Set the threshold value for accuracy $(\delta_I)$ and the
		maximum number of iterations $(N_{max})$.
		\State \textbf{Initialization:}
		Initialize $\mathbf{x}^{[i]}$ with a feasible initialization point and set $i = 0$.
		\While {$\left |r_{sec}^{[m+1]}-r_{sec}^{[i]} \right | \geq \delta_I$ or $i \leq N_{max}$} \textbf{(I)-(III)}
		\State    \textbf{I:}  
		Find $\mathbf{x}^{[i+1]}$ by solving \eqref{eq-alg1}.
		\State    \textbf{II:} Update the slack variables based on $\mathbf{x}^{[i+1]}$.
		\State  \textbf{III:} $i=i+1$.
		\EndWhile
		\State Output: $\mathbf{P}^*$, $\mathbf{F}^*$.
	\end{algorithmic}
\end{algorithm}

Note that if Algorithm \ref{Alg1} were to be
initialized randomly, it may fail initially due to infeasibility. To circumvent this
issue, we further provide a feasible initial point search  (FIPS) in Algorithm \ref{Alg2}. In this approach, we minimize an infeasibility indicator parameter $s > 0$ as a measure of the violation of the constraints in \eqref{eq-alg1}. Thus, we rewrite the constraints of \eqref{eq-alg1} in the form of $\mathcal{G}_i  |_{i=1}^{11} \leq s$,
where $\mathcal{G}_i$ represents the $i$-th constraint with all terms shifted to the left-hand side, and we formulate the feasibility problem as follows:
\begin{equation} \label{eq-alg2}
\min_{\mathbf{x}} \ s  \; \quad \mathrm{s.t.}\quad \;  \mathcal{G}_i  |_{i=1}^{11} \leq s,
\end{equation}	

\begin{algorithm}
	\caption{FIPS Algorithm}
	\label{Alg2}
	\begin{algorithmic}[1]
		\State \textbf{Input:}
		Set the threshold value for accuracy $(\delta_e)$ and the
		maximum number of iterations $(M_{max})$.
		\State \textbf{Initialization:}
		Choose a random initialization $\mathbf{x}^{[i]}$ and set $m = 0$.
		\While {$\left |s^{[i+1]}-s^{[i]} \right | \geq \delta_e$ or $i \leq M_{max}$} \textbf{(I),(II)}
		\State    \textbf{I:}  
		Solve \eqref{eq-alg2}.
		\State  \textbf{II:} $i=i+1$.
		\EndWhile
		\State Output: $\mathbf{x}^*$.
	\end{algorithmic}
\end{algorithm}

The above initialization procedure was previously proposed in \cite{majid-hamed,majid-hamed2} as
a low-complexity scheme for efficiently finding feasible initial
points. Overall, the proposed FIPS Algorithm 2 runs as the
first step, and then the calculated initial points (IPs) are fed to
Algorithm \ref{Alg1}. Algorithm \ref{Alg2} commences with random IPs and the algorithm stops if either the stopping criterion is satisfied or the maximum number of iterations is reached.
We remark that if at the $i$-th iteration the current objective
value s is zero, the algorithm stops even if the other stopping
criteria are not satisfied.

\subsection{SPCA-based Solution for Transmission and Reflection Beamforming Optimization}

For the next step, in order to simplify the optimization problem, we change the objective function to optimize the sum rate for the legitimate users. Therefore, for any given precoder vectors  $\mathbf{p}_c$, $\mathbf{p}_k$, and $\mathbf{f}_j$, the optimization problem for the RIS beamforming vector $\mathbf{u}_n$ is reformulated as

\begin{subequations} \label{P4}
		\allowdisplaybreaks
	\begin{align} 
	&\max_{\mathbf{u}_p}\left( R_c + \sum_{k} \gamma_{k} \right),\\
	\mathrm{s.t.}\ & R_c \leq R_{c,k} ,\ \forall k \in \mathcal{K}, \label{P4-b}\\
	& \gamma_k \leq R_{k} ,\ \forall k \in \mathcal{K}, \label{P4-c}\\
	&\alpha_k R_c \geq r_c, \ \forall k \in \mathcal{K}, \label{P4-d}\\
	&\sum_{j \in \mathcal{J}} \min_{\mathbf{\Delta g}_{r,j}} \{Q_j\} \geq E^{th},\label{P4-e}\\
	& \beta_m^p, \theta_m^p \in  \mathbb{R}_{\beta,\theta}, \ \forall m \in \mathbb{M}, \ \forall p \in \{t,r\}. \label{P4-f} 
	\end{align}
\end{subequations}

Given the precoders $\mathbf{p}_c$, $\mathbf{p}_k$, and $\mathbf{f}_j$, we denote 
$\bar{\mathbf{h}}_{t,k,n} = \mathrm{diag}(\mathbf{g}_{t,k}^H) \mathbf{H} \mathbf{p}_n $, for $n \in \{c,k\}$, $\bar{\mathbf{h}}_{t,k,j}= \mathrm{diag}(\mathbf{g}_{t,k}^H) \mathbf{H} \mathbf{f}_j $, $\mathbf{v}_p = \mathbf{u}_p^H$, and $\mathbf{V}_p =\mathbf{v}_p \mathbf{v}_p^H$, where  $\mathbf{V}_p \succeq 0$, $\mathrm{rank}(\mathbf{V}_p)=1$, and $[\mathbf{V}_p]_{m,m}=\beta_m^p, p \in \{t,r\}$. Hence, we have:
\begin{align} 
	&\left | \mathbf{h}_{t,k} \mathbf{p}_n \right |^2 = \left | \mathbf{v}_t^H \bar{\mathbf{h}}_{t,k,n} \right |^2 = \mathrm{Tr} (\mathbf{V}_t \bar{\mathbf{H}}_{t,k,n}), \ n \in \{c,k\}, \\
	&\left | \mathbf{h}_{t,k} \mathbf{f}_j \right |^2 = \left | \mathbf{v}_t^H \bar{\mathbf{h}}_{t,k,j} \right |^2 = \mathrm{Tr} (\mathbf{V}_t \bar{\mathbf{H}}_{t,k,j}), 
\end{align}
where $\bar{\mathbf{H}}_{t,k,n}=\bar{\mathbf{h}}_{t,k,n}\bar{\mathbf{h}}_{t,k,n}^H, \ n \in \{c,k,j\}$.
Before solving problem \eqref{P4}, we also introduce a slack variable set $\{ A_{c,k}, B_{c,k}, A_{k}, B_{k} | k \in \mathcal{K}\}$ defined as
	\allowdisplaybreaks
\begin{align} 
\frac{1}{A_{c,k}} & =\mathrm{Tr} (\mathbf{V}_t \bar{\mathbf{H}}_{t,k,c}), \label{Ack}\\
 B_{c,k} & = \sum_{k' \in \mathcal{K} }\mathrm{Tr} (\mathbf{V}_t \bar{\mathbf{H}}_{t,k,k'})+\sum_{j \in \mathcal{J} }\mathrm{Tr} (\mathbf{V}_t \bar{\mathbf{H}}_{t,k,j})+1, \label{Bck} \\
 \frac{1}{A_{k}} & =\mathrm{Tr} (\mathbf{V}_t \bar{\mathbf{H}}_{t,k,k}), \label{Ak}\\
 B_{k} & = \sum_{k' \in \mathcal{K}, k' \neq k }\mathrm{Tr} (\mathbf{V}_t \bar{\mathbf{H}}_{t,k,k'})+\sum_{j \in \mathcal{J} }\mathrm{Tr} (\mathbf{V}_t \bar{\mathbf{H}}_{t,k,j})+1. \label{Bk}
\end{align}

Upon substituting \eqref{Ack} and \eqref{Bck} into \eqref{gam-ck} and  \eqref{Ak} and \eqref{Bk} into \eqref{gam-k}, the achievable data rates for the common and private streams, respectively, can be rewritten as
\begin{align} 
R_{c,k}&=\mathrm{log}_2 \left( 1+ \frac{1}{A_{c,k}B_{c,k}}  \right), \label{Rck2}\\
R_{k}&=\mathrm{log}_2 \left( 1+ \frac{1}{A_{k}B_{k}}  \right). \label{Rk2}
\end{align}

The surrogate constraints \eqref{Rck2} and \eqref{Rk2} to be used in place of \eqref{P4-b} and \eqref{P4-c} are still non-convex. However, $\mathrm{log}_2 \left( 1+ \frac{1}{xy}  \right)$ is a joint convex function with respect to $x$ and
$y$, so using a first-order Taylor expansion we approximate the right-hand side of \eqref{Rck2} and \eqref{Rk2} by the following lower bounds:
\begin{align} 
 \mathrm{log}_2 \left( 1+ \frac{1}{A_{c,k}B_{c,k}}  \right)  \geq \tilde{R}_{c,k} =
\mathrm{log}_2 \left( 1+ \frac{1}{A_{c,k}^{[l]}B_{c,k}^{[l]}}  \right)\nonumber\\ - \frac{\mathrm{log}_2(e)(A_{c,k}-A_{c,k}^{[l]})}{A_{c,k}^{[l]}+(A_{c,k}^{[l]})^2 B_{c,k}^{[l]}} -  \frac{\mathrm{log}_2(e)(B_{c,k}-B_{c,k}^{[l]})}{B_{c,k}^{[l]}+(B_{c,k}^{[l]})^2 A_{c,k}^{[l]}} \label{Rc-tild}\\
 \mathrm{log}_2 \left( 1+ \frac{1}{A_{k}B_{k}}  \right)  \geq \tilde{R}_{k} =
 \mathrm{log}_2 \left( 1+ \frac{1}{A_{k}^{[l]}B_{k}^{[l]}}  \right) \nonumber\\- \frac{\mathrm{log}_2(e)(A_{k}-A_{k}^{[l]})}{A_{k}^{[l]}+(A_{k}^{[l]})^2 B_{k}^{[l]}} -  \frac{\mathrm{log}_2(e)(B_{k}-B_{k}^{[l]})}{B_{k}^{[l]}+(B_{k}^{[l]})^2 A_{k}^{[l]}},\label{R-tild}
 \end{align}	
where $A_{c,k}^{[l]}$, $B_{c,k}^{[l]}$, $A_{k}^{[l]}$, and $B_{k}^{[l]}$ represent the values of $A_{c,k}$, $B_{c,k}$, $A_{k}$, and $B_{k}$ in the $l$-th iteration, respectively. Thus, the transmission and reflection beamforming optimization problem in \eqref{P4} assuming fixed
precoders can be reformulated as
\begin{subequations} \label{P5}
		\allowdisplaybreaks
	\begin{align} 
	&\max_{\mathbf{v}_p}\left( R_c + \sum_{k} \gamma_{k} \right),\\
	\mathrm{s.t.}\ &\frac{1}{A_{c,k}} \leq \mathrm{Tr} (\mathbf{V}_t \bar{\mathbf{H}}_{t,k,c}), \label{P5-b}\\	
	& B_{c,k}  \geq \sum_{k' \in \mathcal{K} }\mathrm{Tr} (\mathbf{V}_t \bar{\mathbf{H}}_{t,k,k'})+\sum_{j \in \mathcal{J} }\mathrm{Tr} (\mathbf{V}_t \bar{\mathbf{H}}_{t,k,j})+1, \label{P5-c} \\
	&  \tilde{R}_{c,k} \geq R_c, \label{P5-d} \\	
	& \frac{1}{A_{k}}  =\mathrm{Tr} (\mathbf{V}_t \bar{\mathbf{H}}_{t,k,k}), \label{P5-e}\\
	& B_{k}  = \sum_{k' \in \mathcal{K}, k' \neq k }\mathrm{Tr} (\mathbf{V}_t \bar{\mathbf{H}}_{t,k,k'})+\sum_{j \in \mathcal{J} }\mathrm{Tr} (\mathbf{V}_t \bar{\mathbf{H}}_{t,k,j})+1,\label{P5-f} \\	&  \tilde{R}_{k} \geq R_k, \label{P5-g} \\	&
	  \beta_m^t+\beta_m^r=1, \label{P5-h}\\ &
	  [\mathbf{V}_p]_{m,m}=\beta_m^p, \label{P5-i}\\ &
	  \mathbf{V}_p \succeq 0, \label{P5-j}\\ &
	  \mathrm{rank}(\mathbf{V}_p)=1, \label{P5-k} \\ & 
	  \eqref{P4-d},\eqref{P4-e},	  
	\end{align}
\end{subequations}
where $p \in \{t,r\}$, and $m \in \mathbb{M}$.

As in \cite{rank1,rank2}, the non-convex rank-one constraint \eqref{P5-k} can be replaced by the following relaxed convex constraint:
\begin{align} 
\epsilon_{max}(\mathbf{V}_p) \geq \epsilon^{[l]} \mathrm{Tr}
(\mathbf{V}_p),
\end{align}
where $\epsilon_{max}(\mathbf{V}_p)$ denotes the maximum eigenvalue of matrix $\mathbf{V}_p$, and $\epsilon^{[l]}$ is a relaxation parameter in the $l$-th iteration that scales the ratio of $\epsilon_{max}(\mathbf{V}_p)$ to the trace of $\mathbf{V}_p$. Specifically, $\epsilon^{[l]}=0$ indicates that the rank-one constraint is dropped, while $\epsilon^{[l]}=1$ means it is retained. Therefore, we can increase $\epsilon^{[l]}$ from 0 to 1 with each iteration to gradually approach a rank-one solution. It is noted that $\epsilon_{max}(\mathbf{V}_p)$ is not differentiable, so we use the following expression to approximate it:
\begin{align} 
\epsilon_{max}(\mathbf{V}_p) \approx  e_{max}^H(\mathbf{V}_p^{[l]})\mathbf{V}_p e_{max}(\mathbf{V}_p^{[l]}),
\end{align}
where $e_{max}(\mathbf{V}_p^{[lu]})$ is the eigenvector corresponding to the maximum eigenvalue of $\mathbf{V}_p^{[l]}$.
Thus, solving problem \eqref{P5} is transformed to solving the following relaxed problem:
\begin{subequations} \label{P6}
		\allowdisplaybreaks
	\begin{align}
	&\max_{\mathbf{u}_p}\left( R_c + \sum_{k} \gamma_{k} \right),\\
	\mathrm{s.t.}\ &
	 e_{max}^H(\mathbf{V}_p^{[l]})\mathbf{V}_p e_{max}(\mathbf{V}_p^{[l]}) \geq \epsilon^{[l]} \mathrm{Tr}
	(\mathbf{V}_p), \label{P6-b}\\ & 
	  \eqref{P4-d},\eqref{P4-e}, \eqref{P5-b}-\eqref{P5-j}.	  
	\end{align}
\end{subequations}
The parameter $\epsilon^{[l]}$ can be updated via \cite{rank2}
\begin{align}\label{eps-updat}
 \epsilon^{[l+1]} = \min\left(1,\frac{\epsilon_{max}(\mathbf{V}_p^{[l]}) }{\mathrm{Tr}
 	(\mathbf{V}_p^{[l]})} + \Delta^{[l]}\right),
\end{align}
where $ \Delta^{[l]}$ is the step size.

Next we address the non-convex constraints \eqref{P4-d} and \eqref{P4-e}. According to Proposition 5, the affine approximation of constraint \eqref{P4-d} is obtained by \eqref{eq-23}. In order to deal with the non-convexity of  \eqref{P4-e}, we first modify the expression for some variables as follows:
\begin{align} 
&\left | \mathbf{h}_{r,j} \mathbf{p}_n \right |^2 = \left | \mathbf{v}_r^H \bar{\mathbf{h}}_{r,j,n} \right |^2 = \mathrm{Tr} (\mathbf{V}_r \bar{\mathbf{H}}_{r,j,n}),  \\
&\left | \mathbf{h}_{r,j} \mathbf{f}_j \right |^2 = \left | \mathbf{v}_r^H \bar{\mathbf{h}}_{r,j,j'} \right |^2 = \mathrm{Tr} (\mathbf{V}_r \bar{\mathbf{H}}_{r,j,j'}), 
\end{align}
where $\bar{\mathbf{h}}_{r,j,n}= \mathrm{diag}(\mathbf{g}_{r,j}^H)\mathbf{H} \mathbf{p}_n, \ n \in \{c,k\}$, $\bar{\mathbf{h}}_{r,j,j'}= \mathrm{diag}(\mathbf{g}_{r,j}^H)\mathbf{H} \mathbf{f}_j'$, 
$\bar{\mathbf{H}}_{r,j,n}=\bar{\mathbf{h}}_{r,j,n}\bar{\mathbf{h}}_{r,j,n}^H$, and $\bar{\mathbf{H}}_{r,j,j'}=\bar{\mathbf{h}}_{r,j,j'}\bar{\mathbf{h}}_{r,j,j'}^H$. Then, applying the same approach as in Proposition 6 and after some simple algebraic manipulations,  $\eqref{P4-e}$ can be rewritten as:
\begin{align}  \label{eq-44}
\begin{cases}
(\RN{1}):\displaystyle\sum_{j \in \mathcal{J}} \left(\lambda_{j,c} + \displaystyle\sum_{k \in \mathcal{K}} \lambda_{j,k} + \displaystyle\sum_{j' \in \mathcal{J}} \xi_{j,j'}
\right) \geq E^{th},
\\
(\RN{2}): \mathrm{Tr} \left[(\bar{\mathbf{H}}_{r,j,c}-\mu_c \mathbf{I}) \mathbf{V}_{r} \right]\geq \lambda_{j,c},  \ \forall j \in \mathcal{J}, 
\\
(\RN{3}): \mathrm{Tr} \left[(\bar{\mathbf{H}}_{r,j,k}-\mu_k \mathbf{I}) \mathbf{V}_{r}\right] \geq \lambda_{j,k}, \ \forall j \in \mathcal{J}, k \in \mathcal{K},   
\\
(\RN{4}):\mathrm{Tr} \left[(\bar{\mathbf{H}}_{r,j,j'}-\mu_j \mathbf{I}) \mathbf{V}_{r} \right]\geq \xi_{j,j'},
\ \forall j,j' \in \mathcal{J},
\end{cases}
\end{align}
where $\mu_n = \nu^2 + 2\nu \Vert \bar{\mathbf{h}}_{r,j,n} \Vert_2, \ n \in \{c,k,j\}$. Therefore, problem \eqref{P5} can finally be written as
\begin{align} \label{P7}
	&\max_{\mathbf{u}_p}\left( R_c + \sum_{k} \gamma_{k} \right),\nonumber\\
	\mathrm{s.t.}\ &  \eqref{P6-b}, \eqref{eq-23}, \eqref{eq-44}, \eqref{P5-b}-\eqref{P5-j} .
\end{align}
Problem \eqref{P7} is a standard convex SDP, which can be solved efficiently by numerical solvers
such as the SDP tool in CVX \cite{CVX}. 
The details of the proposed sequential constraint relaxation algorithm are presented in Algorithm \ref{Alg3}. 

\begin{algorithm}
	\caption{SPCA-based Algorithm for obtaining \{$\mathbf{V}_p^*$\}}
	\label{Alg3}
	\begin{algorithmic}[1]
		\State Choose initial
		feasible points $\mathbf{V}_p^{[0]}$, step size $\Delta^{[0]}$, error tolerance $\delta_p$, and set the relaxation parameter $\epsilon^{[l]}=0$
		and the iteration index $l = 0$.
		\Repeat 
		\State    Solve problem \eqref{P7} to obtain $\mathbf{V}_p$;
        \If{problem \eqref{P7} is solvable}
        \State Update $\mathbf{V}_p^{[l+1]}=\mathbf{V}_p$;
        \State Update $\Delta^{[l+1]}=\Delta^{[l]}$;
        \Else
        \State  Update $\Delta^{[l+1]}=\Delta^{[l]}/2$;
        \EndIf
		\State   $l=l+1$.
	\State Update $\epsilon^{[l+1]}$ via \eqref{eps-updat};
	\Until{$\left| 1- \epsilon^{[l]}\right| \leq \delta_p$ and the objective value of problem \eqref{P7} converges.} 
		\State Output: $\mathbf{V}_p^*$.
	\end{algorithmic}
\end{algorithm}

\section{Overall Proposed Algorithm, Convergence and Complexity}\label{Overal method}
Based on the above discussions, here we describe     the proposed two-step iterative
algorithm to solve the original problem \eqref{P1} in Algorithm \eqref{Alg4}. Our approach
is to alternate between optimizing the precoders and optimizing the transmission and reflection beamforming vectors.

\begin{algorithm}
	\caption{The overall proposed Algorithm}
	\label{Alg4}
	\begin{algorithmic}[1]
		\State \textbf{Input:}
		Set the parameters $\delta_e$, $\delta_p$, $\delta_I$, $N_{max}$, $M_{max}$, $\Delta^{[0]}$.
		\State \textbf{Initialization:} Randomly initialize 
		$\mathbf{x}^{[0]}$, $\mathbf{V}_p^{[0]}$, and set the relaxation parameter $\epsilon^{[l]}=0$ and iteration indices $l = 0$, $i = 0$.
		\Repeat 
		\State    Update $\mathbf{x}^{[i]}$  via Algorithm \ref{Alg2};
	\State    Update the  precoder vectors $\mathbf{P}$, $\mathbf{F}$  via Algorithm \ref{Alg1};
	\State    Update the beamforming vector $\mathbf{u}_p$ via Algorithm \ref{Alg3};			
		\Until{the objective value of problem \eqref{P1} converges.} 
		\State Output: $\mathbf{P}^*$, $\mathbf{F}^*$, $\mathbf{V}_p^*$.
	\end{algorithmic}
\end{algorithm}

\subsection{Convergence Analysis} 
In this section we provide a convergence analysis of the proposed
SPCA algorithm. Since the original problem \eqref{P1} is non-convex,
it is not possible to prove convergence to a global minimum, but
convergence to Karush-Kuhn-Tucker (KKT) points under some regularity conditions
may be shown. The following lemmas, referenced from \cite{Beck, Bastami2023}, will be used in the convergence proof. For simplicity we define $\Omega$ to be the feasible set of \eqref{P1}, and
$\Omega^{[t]}$ to be the feasible set of \eqref{eq-alg1} for the $t^{th}$ iteration.
\begin{Lemma}
    Let $\mathcal{D} : \mathbb{R}^n \rightarrow \mathbb{R}$ be a strictly convex and differentiable function on a nonempty convex set $S \subseteq \mathbb{R}^n$. Then D is
strongly convex on the set S. 
\end{Lemma}
\begin{proof}
	See \cite{Beck}.
\end{proof}
\begin{Lemma}
    
Let $\{\mathbf{x}^{[t]}\}$ be the sequence generated by the SPCA method.
Then, for every $t \geq 0$: $\mathbf{i})$ $\Omega^{[t]}
\subseteq \Omega$, $\mathbf{ii})$ $\mathbf{x}^{[t]} \in \Omega^{[t]} \cap \Omega^{[t+1]}$, 
$\mathbf{iii})$ $\{\mathbf{x}^{[t]}\}$ is a feasible point of \eqref{P1}, $\mathbf{iv})$, $r_{sec}^{[t]} \leq r_{sec}^{[t+1]}$. 
\end{Lemma}
\begin{proof}
	See \cite{Beck}.
\end{proof}
\begin{Lemma}
The sequence $r_{sec}^{[t]}$ converges.
\end{Lemma}
\begin{proof}
	See \cite{Beck,Bastami2023}.
\end{proof}
Given the aforementioned lemmas and following the same approach as in \cite{Bastami2023}, the convergence of Algorithm \ref{Alg1} is trivial. The convergence of Algorithm \ref{Alg3} can be proved similarly to the algorithm in \cite{Cao-Poor}. The Algorithm \ref{Alg4} iterates between Algorithm \ref{Alg1} and Algorithm \ref{Alg3}; thus, its convergence follows as well.

\subsection{Complexity Analysis} 
The complexity of Algorithm \eqref{Alg4} depends primarily on the complexity of
Algorithms \eqref{Alg1} and \eqref{Alg3}.
In Algorithm \eqref{Alg1}, the major complexity lies in the matrix multiplication of \eqref{lem1-formul}. Thus, the worst-case complexity order of solving the convex problem is given
by $O_1 \triangleq \mathcal{O}\left( N_{max} (M^2 N_T)^2 \max(K,J)  \right)$, where $N_{max}$ is the number of iterations.

Algorithm \eqref{Alg3} for finding the transmission and reflection beamforming vectors in \eqref{P7} requires on the order of $O_2 \triangleq \mathcal{O} \left(l_{max} \max(M,2(K+J))^4 \sqrt{M} \log_2 \frac{1}{\epsilon_1}  \right)$ computations, where $l_{max}$ is the number of iterations for Algorithm \eqref{Alg3} and $\epsilon_1$ is the solution accuracy.
 The overall complexity of Algorithm \eqref{Alg4} is $\mathcal{O}(W_{max}(O_1 + O_2))$ where $W_{max}$ is the number of iterations required for the entire algorithm.

\section{Simulation Results}\label{Simulation Results}
In this section, we present some numerical results for
our proposed framework. The simulation setting for the STAR-RIS-RSMA system is based on
the following scenario, unless otherwise stated.  The BS
is located $20 m$ above the ground i.e. $(0,0,20)$, and the RIS is located at $(0,30,20)$. There are $K=3$ IRs and $J=3$ UERs randomly located on    the two different sides of the RIS, as depicted in Fig. \ref{sys_simu}
\begin{figure}
	\centering
	\includegraphics[width=1\linewidth]{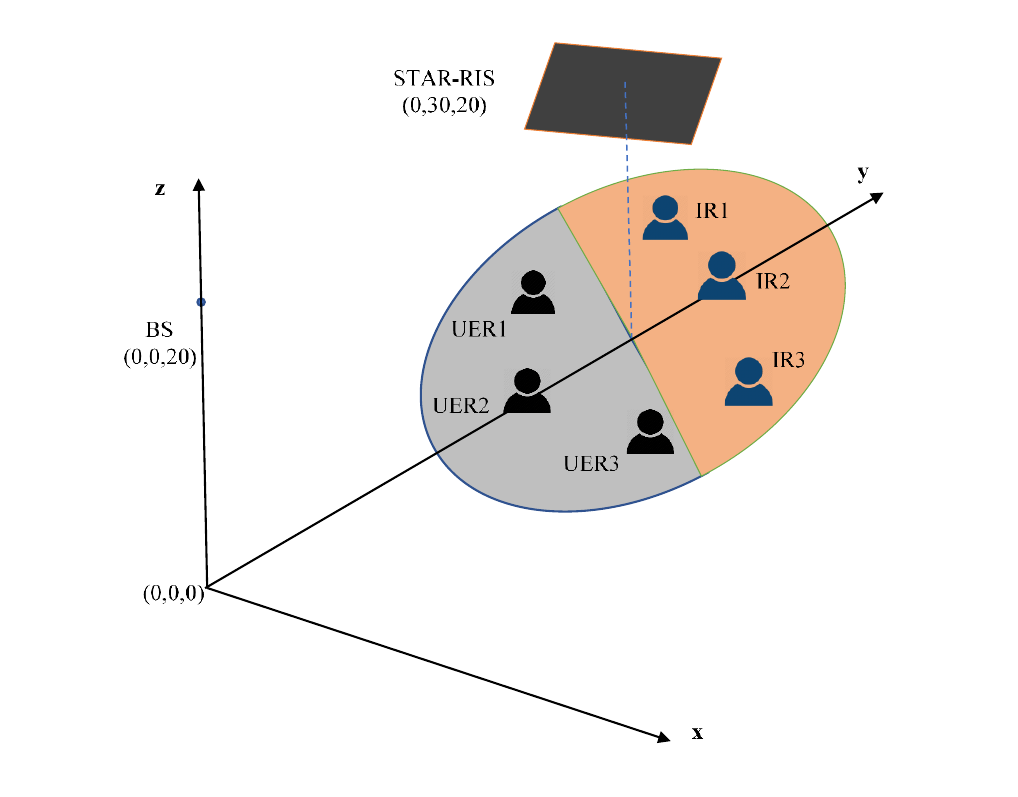}
	\caption{Simulation scenario for the STAR-RIS-RSMA system.}
	\label{sys_simu}
\end{figure}

The distance-dependent path loss for node $n$ is modeled
as $\mathrm{PL}_n \triangleq d_n ^ {-\alpha_n}$ where $d_n$ is the distance between the RIS and the user, and the path loss exponent $-\alpha_n$ is taken from the probabilistic model of \cite{path-loss}, which is appropriate for low-altitude RISs with a LoS link that may be blocked due to environmental obstacles. In particular, the model assumes
\begin{align}
\alpha_n \triangleq \frac{\mathcal{L}_n-\mathcal{N}_n}{1+\lambda_1 e^{\lambda_2(\phi_n - \lambda_1 )} }+ \mathcal{N}_n,
\end{align}
where $\alpha_n$ is composed of both LoS and non-LoS components $\mathcal{L}$ and $\mathcal{N}$, respectively, $\phi_n \triangleq \frac{180}{\pi} \sin^{-1}( \frac{H_{\mathrm{RIS}}}{d_n})$ denotes the elevation angle between the RIS
and the $n$-th user at an aerial-to-ground distance $d$, while $\lambda_1$ and $\lambda_2$
are determined by the wireless environment, e.g., urban, dense
urban, etc.  In our simulations we choose LoS and NLoS path-loss exponents of $\mathcal{L}=2$ and $\mathcal{N}=3.5$, respectively. 
The simulation results are obtained by averaging the performance over 100 channel realizations.
In the algorithmic implementation, we set $N_{max} = M_{max} = 30$ and the maximum threshold value as $\delta_I = 10^{-2}$. 
A channel estimation error with variance
of $\nu  = 10^{-4}$ is assumed for the communication channel $\mathbf{g}_{r,j}$, 
the additive noise at the receiver side is considered to have a
normalized power of 0 dBm, the minimum harvested energy needed by the UERs is $30 \mu W$, and the minimum required number of common stream transmissions is $r_c = 1$.

The average convergence behavior of the proposed SPCA-based approach (Algorithm \ref{Alg1}) is shown in Fig. \ref{fig3}.
The curves show the  secrecy rate achieved by the algorithm versus the number of iterations, and we see that increasing the transmit power leads to an improved secrecy rate at the cost of a longer convergence time. The secrecy rate improves since at some point (between 20-25 dBm) the UERs no longer need additional energy, and the SPCA method can then use the increased transmit power for improving  the security of the IRs.
\begin{figure}
	\centering
	\includegraphics[width=1\linewidth]{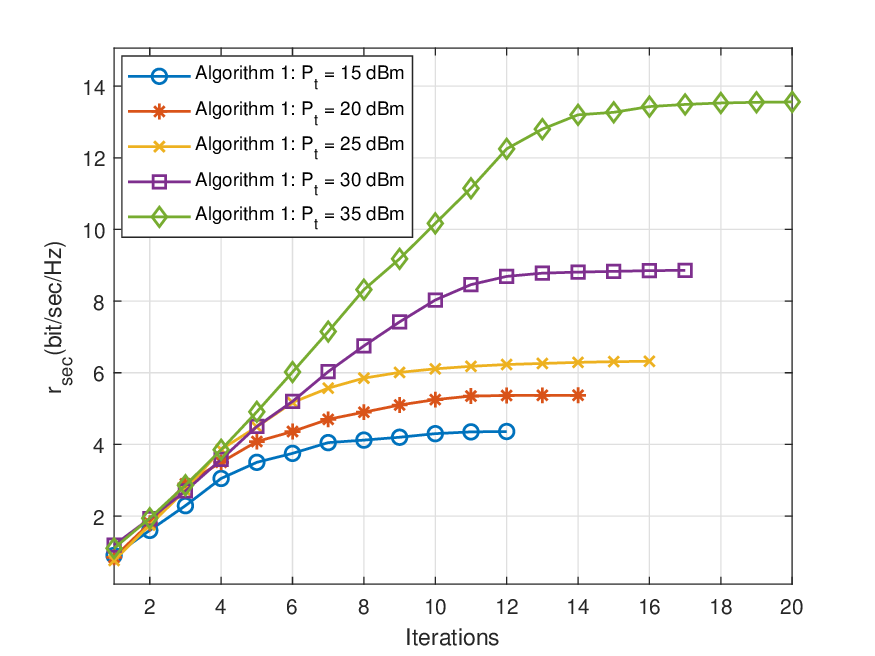}
	\caption{Convergence behavior of Algorithm \ref{Alg1}  versus
		the number of iterations for different values of transmit power when $N_T = 4$, and $M=10$.}
	\label{fig3}
\end{figure}

In Fig. \ref{fig4} we consider the performance of Algorithm \ref{Alg1} for different values of $N_T$. As expected, increasing the number of transmit antennas in turn increases the number of DoFs and we achieve a higher $r_{sec}$.
\begin{figure}
	\centering
	\includegraphics[width=1\linewidth]{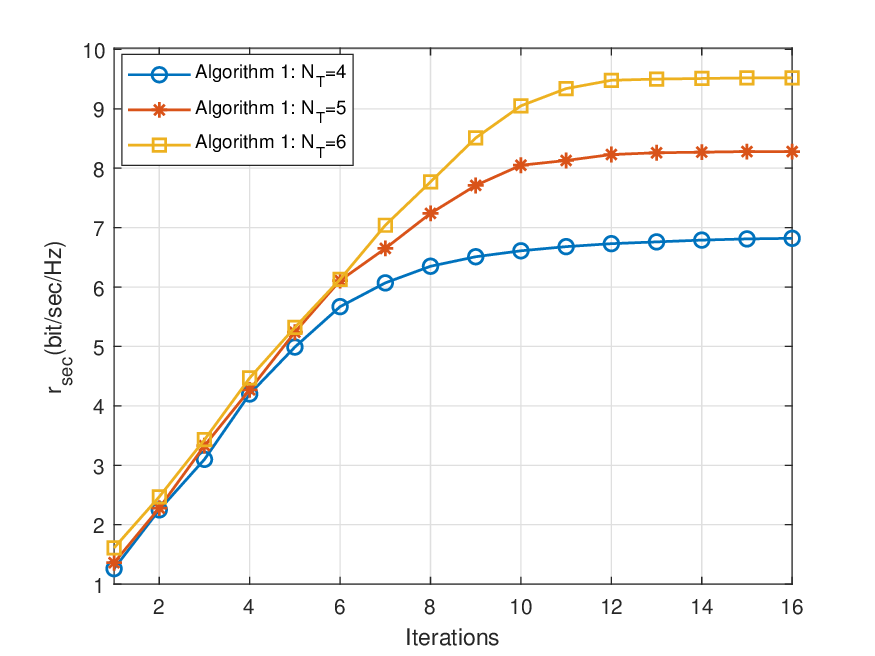}
	\caption{Convergence behavior of Algorithm \ref{Alg1} for different values of $N_T$, when $P_t = 25dBm$ and $M=10$.}
	\label{fig4}
\end{figure}
Fig. \ref{fig5} investigates
the convergence of Algorithm \ref{Alg3}, and we see that the number of iterations required
for attaining convergence increases with $M$, since more transmission and reflection
coefficients have to be optimized. Moreover, we see that increasing $M$ leads to a significant improvement in the achievable sum rate.
\begin{figure}
	\centering
	\includegraphics[width=1\linewidth]{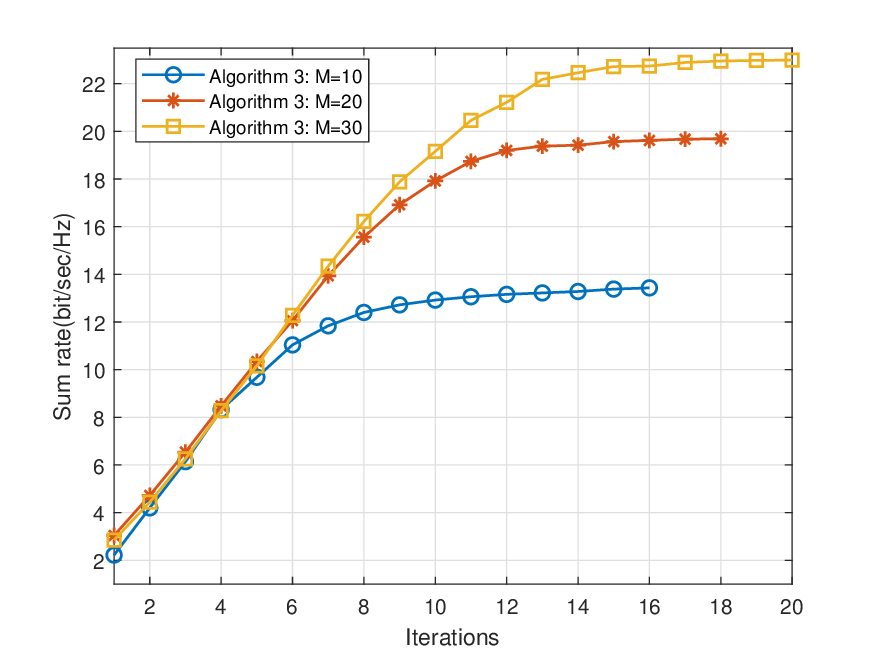}
	\caption{Convergence behavior of Algorithm \ref{Alg3} for different values of $M$, when $P_t = 25dBm$ and $N_T = 4$.}
	\label{fig5}
\end{figure}
The convergence of the overall two-step 
Algorithm \ref{Alg4} is characterized in Fig. \ref{fig6}, which shows rapid convergence in a small number of iterations, confirming
the effectiveness of our proposed algorithm. Moreover, in high transmit power scenarios, $r_{sec}$ increases significantly.
\begin{figure}
	\centering
	\includegraphics[width=1\linewidth]{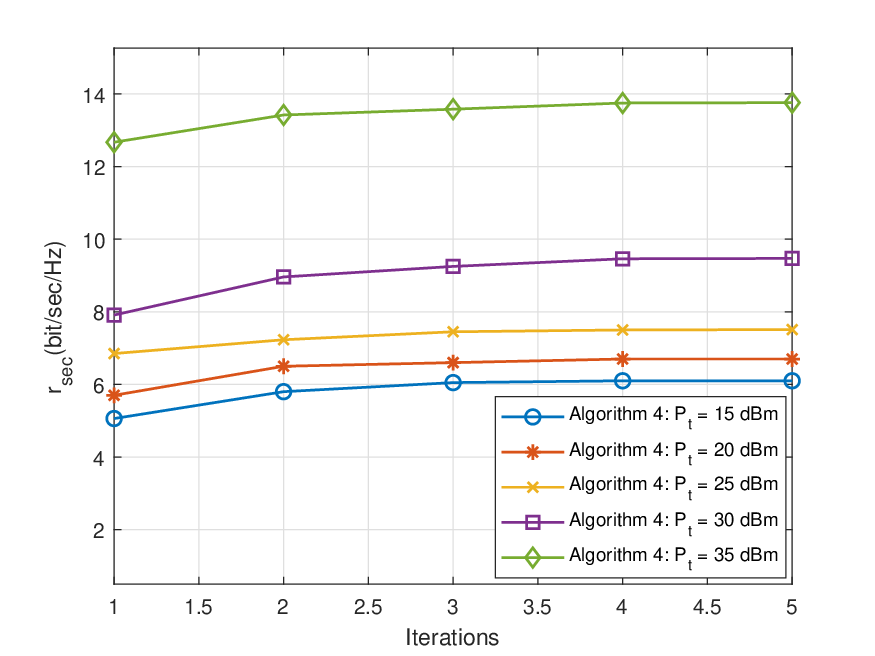}
	\caption{Convergence behavior of Algorithm \ref{Alg4} versus the number of iterations for different values of transmit power when $N_T = 4$, and $M=10$.}
	\label{fig6}
\end{figure}

\begin{figure}
	\centering
	\includegraphics[width=1\linewidth]{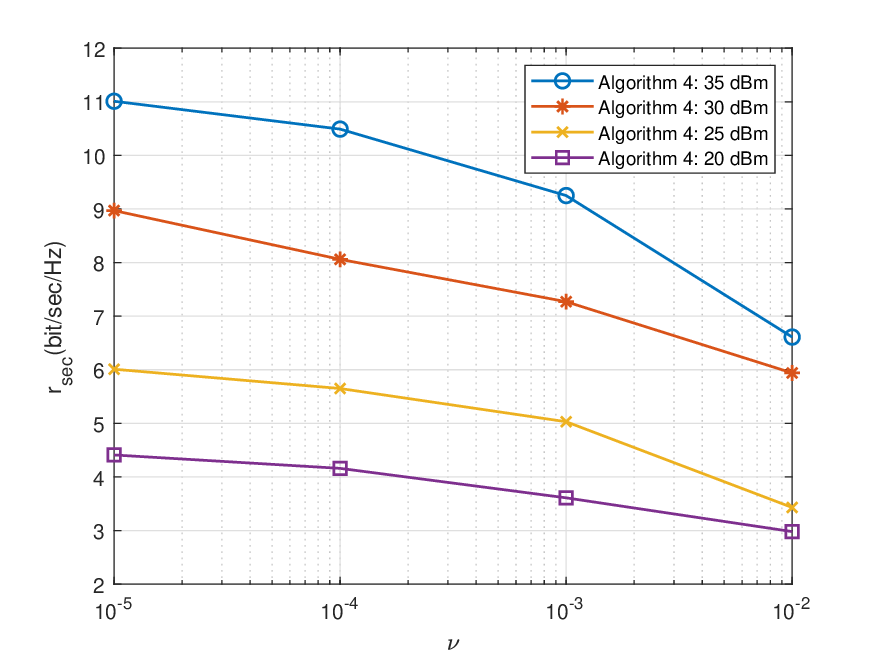}
	\caption{Convergence behavior of Algorithm \ref{Alg4} versus $\nu$ parameterized by $P_t$ when $N_T = 4$, and $M=10$.}
	\label{fig7}
\end{figure}
Fig. \ref{fig7} illustrates the robustness of the proposed framework against
imperfect CSIT, depicting the average worst case secrecy rate versus different levels of CSIT
estimation error $\nu$ and transmit powers $P_t$. Observe that increasing $\nu$ reduces $r_{sec}$ in all cases, although the loss can be offset by increasing $P_t$.
\begin{figure}
	\centering
	\includegraphics[width=1\linewidth]{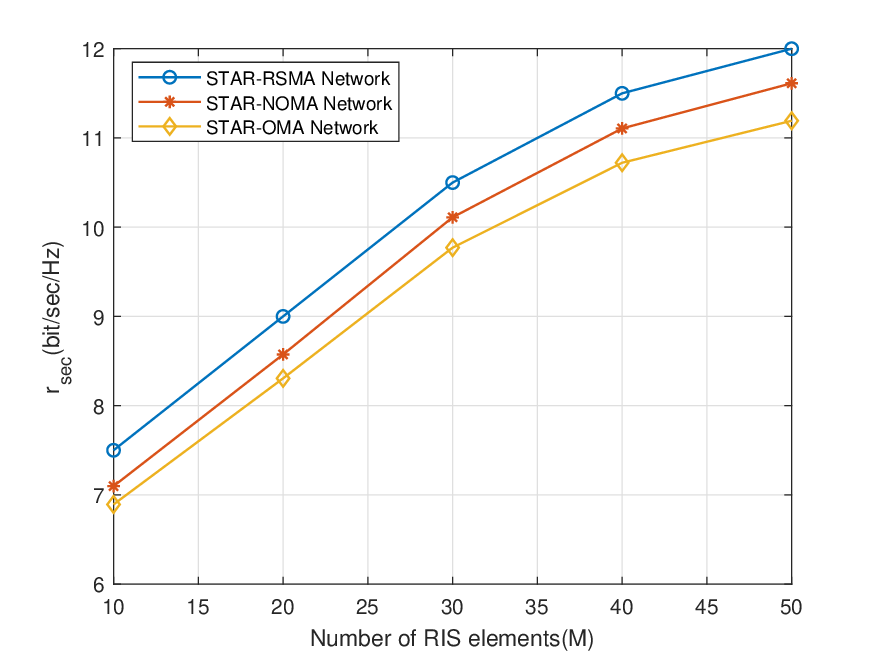}
	\caption{Comparison of the proposed method with NOMA and OMA in STAR-RIS network.}
	\label{fig8}
\end{figure}
In Fig. \ref{fig8} we simulate the same network topology, varying only the multiplexing methods among RSMA, NOMA, and OMA. The results clearly demonstrate that the proposed strategy surpasses both NOMA and OMA, aligning with the expectations set forth in 
\cite{RSMA}.

Finally, we illustrate the execution time duration (in seconds) of the overall proposed algorithm (Algorithm \ref{Alg4}) for different numbers of users. 
 An Intel core i5, 4 GHz processor was used and the results are presented in Table \ref{table1}.
\begin{table}[!t]
	\small
	\renewcommand{\arraystretch}{1.3}
	\caption{comparison of run-time of Algorithm \ref{Alg4} for different number of users}
	\centering
	\begin{tabular}{c c c}
		\hline\hline
		$K$ (N.O of IRs)
		& $J$ (N.O of UERs)  & run-time (s)\\
		\hline
     	1  & 1&  65s \\
		2 & 1   & 82s 	\\
		2 & 2   &  130s		\\
		3 & 2   &  	168s	\\
		3 & 3   &  	209s	\\
		\hline
	\end{tabular}
	\label{table1}
\end{table}
Observe from Table \ref{table1}  that  increasing the number of IRs and 
run-time of Algorithm \ref{Alg4}. For instance, the run-time for $K=2$ and $J=2$ is $130s$, which is approximately twice that of $K=1$ and $J=1$. This is because the number of constraints in Algorithm \ref{Alg4} depends on $K$ and $J$.

\section{Conclusions and Future Work}\label{conc}
RISs are capable of improving next-generation networks. Thus, in this paper, a STAR-RIS-RSMA method was proposed for the simultaneous transmission of information and power to IRs and UERs, respectively.
Since the UERs are able to wiretap the IR streams, our objective is to maximize the sum secrecy rate of the IRs under realistic constraints on the total transmit power and the energy collected by the UERs. This is achieved by jointly optimizing
the precoder vectors and the RIS transmission and reflection
beamforming coefficients. The  formulated problem is non-convex with intricately coupled variables. In order to tackle this challenge a suboptimal two-step iterative algorithm was proposed in which the precoders and the RIS transmission/reflection beamforming coefficients are optimized in an alternating fashion. The results inferred from our  simulations were used to validate the performance of the proposed method. Increasing the transmit power or the number of BS antennas improves the performance of the system. The number of STAR-RIS elements $M$ determines the number of DoFs of the system. Thus, the results demonstrate that when $M$ increases, the proposed method achieves higher rates. Furthermore, the proposed method converges rapidly, achieves a high secrecy rate, and shows a remarkable robustness to imperfect CSI. 
For ease of exposition, the quantitative results are
summarized in Table \ref{table2}.

\begin{table}[!t]
	\small
	\renewcommand{\arraystretch}{1.3}
	\caption{Brief Review on Quantitative Results}
	\centering
	\begin{tabular}{ p{8cm} }
		\hline
		 \begin{itemize}
		 	\item 
		 The  secrecy rate achieved by Algorithm \ref{Alg1} improves with increasing the transmit power but results in a longer convergence time. The secrecy rate improves since for high SNRs, the UERs do not need additional energy, and the SPCA method can then use the increased transmit power to increase the security of the IRs.
		 \end{itemize}
		 \\
		\hline
	\begin{itemize}
		\item Increasing the number of transmit antennas  increases the number of DoFs and we achieve a higher $r_{sec}$.
	\end{itemize} \\
				\hline
	\begin{itemize}
		\item Increasing the number of transmission and reflection
		coefficients leads to a significant improvement in the achievable sum rate as well as more complexity in Algorithm \ref{Alg3} with slow   er convergence.
	\end{itemize} 	\\		\hline
		\begin{itemize}
			\item The  two-step 
			Algorithm \ref{Alg4} converges rapidly in only a few iterations, confirming
			the effectiveness of the proposed approach.
		\end{itemize}		\\		\hline
		\begin{itemize}
			\item  The proposed  method is robust against
			imperfect CSIT, with a predictable loss when the level of CSI increases. However, the loss can be compensated by increasing the transmitted power.
		\end{itemize}		\\		
		\hline
	\end{tabular}
	\label{table2}
\end{table}

In our subsequent research endeavors, we aim to conduct a thorough theoretical analysis of secure STAR-RIS networks. This will involve a deep dive into their core principles and theoretical limits to not only augment our empirical findings but also to establish a comprehensive theoretical framework guiding practical implementations. Furthermore, we plan to explore more realistic RIS models, specifically incorporating quantized reflection coefficients and a non-energy-saving model. Another significant area of future research will be addressing the challenges associated with nonlinear energy harvesting models in UERs. These expanded research directions will enable us to tackle more complex scenarios and contribute to advancing the state-of-the-art in the field.

\ 
\appendices
\numberwithin{equation}{section}
\section{Proof of Proposition 1} \label{app1}
The constraint \eqref{P3-c} can be equivalently rewritten $\forall j \in \mathcal{J}$ as follows:
\begin{align}
R_{c,j}^{\mathrm{UER}}  \leq \alpha_{c,j}, 
\qquad \{\forall \ \mathbf{\Delta g}_{r,j} \in \Theta_g \}.\label{A.1}
\end{align}
Problem \eqref{A.1} represents a search over all possible values of the channel uncertainties to find the worst-case.
On the other hand,
based on what we inferred earlier from \eqref{P3-b} and  \eqref{P3-c}, the goal is to minimize $R_{c,j}^{\mathrm{UER}}$ by minimizing its ceiling rate
$\alpha_{c,j}$. Given
this fact, and with the aim of linearizing \eqref{P3-c}, after exploiting the definition of $R_{c,j}^{\mathrm{UER}}$
and using some variable transformations, the pair of extra auxiliary constraints becomes
\begin{align}
&\gamma_{c,j}^{\mathrm{UER}}  \leq \rho_{c,j}\label{A.2}\\
& 1 + \rho_{c,j} - 2^{\alpha_{c,j}} \leq 0.\label{A.3}
\end{align}
The non-convexity of \eqref{A.2} is revealed by inserting the definition of $\gamma_{c,j}^{\mathrm{UER}}$
as well as utilizing the auxiliary variables introduced in Proposition 1. The non-convexity of \eqref{A.3} results from the term $2^{\alpha_{c,j}}$, and the
non-convex factor is replaced at the $i$-th iteration using a first order Taylor expansion of $2^{\alpha_{c,j}}$ based on the operator $\Gamma^{[i]}(\alpha_{c,j})$. On the other hand, the problem \eqref{A.2} can be recast as follows after some trivial manipulations:
\begin{align}
& \frac{x_{c,j}^2}{\sum_{k' \in \mathcal{K}} a_{j,k'}+\sum_{j' \in \mathcal{J}} b_{j,j'}+1} \leq \rho_{c,j}, \\
& \max_{ \mathbf{\Delta g}_{r,j}} |\mathbf{h}_{r,j}\mathbf{p}_c| \leq x_{c,j}, \label{A.5}\\
& \min_{ \mathbf{\Delta g}_{r,j}} \left|\mathbf{h}_{r,j}\mathbf{p}_{k'}\right|^2 \leq a_{j,k'}, \label{A.6}\\
& \min_{ \mathbf{\Delta g}_{r,j}} \left|\mathbf{h}_{r,j}\mathbf{f}_{j'}\right|^2 \leq b_{j,j'}.\label{A.7}
\end{align}

In order to overcome the non-convexity of  \eqref{A.5}-\eqref{A.7}, we
exploit the combined channel definition $ \mathbf{h}_{r,j} = \mathbf{g}_{r,j}^H \boldsymbol{\Theta}_r \mathbf{H}$, as well as \eqref{g-hat}, and the following proposition.
\begin{Proposition}
For the terms $\mathcal{D}_1(\mathbf{\Delta h})\triangleq \left | \left(\hat{\mathbf{h} }+ \mathbf{\Delta h}\right)^H \mathbf{u} \right | $ and  $\mathcal{D}_2(\mathbf{\Delta h})\triangleq \left | \left(\hat{\mathbf{h} }+ \mathbf{\Delta h}\right)^H \mathbf{u} \right |^2 $ with norm-bounded variable $ \Vert \mathbf{\Delta h} \Vert_2 \leq \sigma$, the following results hold:
\begin{align}
& \max_{\Vert \mathbf{\Delta h} \Vert_2 \leq \sigma} \mathcal{D}_1(\mathbf{\Delta h})= \left| \hat{\mathbf{h}}^H\mathbf{u} \right| +\sigma  \Vert \mathbf{u} \Vert_2, \\
& \min_{\Vert \mathbf{\Delta h} \Vert_2 \leq \sigma} \mathcal{D}_2(\mathbf{\Delta h})= \mathrm{Tr} \left[(\hat{\mathbf{H}}-\mu \mathbf{I}) \mathbf{U}\right], \\
& \max_{\Vert \mathbf{\Delta h} \Vert_2 \leq \sigma} \mathcal{D}_2(\mathbf{\Delta h})= \mathrm{Tr} \left[(\hat{\mathbf{H}}+\mu \mathbf{I}) \mathbf{U}\right],
\end{align}
where $\hat{\mathbf{H}} \triangleq \hat{\mathbf{h}} \hat{\mathbf{h}}^H$, $\mathbf{U} \triangleq \mathbf{u} \mathbf{u}^{H}$, and $\mu \triangleq \sigma^2 + 2\sigma \left \Vert \hat{\mathbf{h}} \right \Vert_2$.
\end{Proposition}
\begin{proof}
Please refer to \cite{proof}
\end{proof}

Using Proposition 1, we can relax the right-side of \eqref{A.5}-\eqref{A.7}, leading to the expressions in (\ref{lem1-formul}-(\RN{3}))-(\ref{lem1-formul}-(\RN{5})).

\section{Proof of Proposition 3} \label{app2}
With the aim of linearizing \eqref{P3-g}, after exploiting the definition of $R_k$ and using some variable transformations, the two extra auxiliary constraints are formulated for $\forall k \in \mathcal{K}$, as follows:
\begin{align}
& 1 + \rho_{k} - 2^ {\gamma_k}  \geq 0\label{B.1}
\\ 
& \gamma_{k}^\mathrm{ID} \geq \rho_k.\label{B.2}
\end{align}
To maintain the generality of the problem, the auxiliary constraints must be defined to ensure that they are active at
the optimum solution. Although \eqref{B.1} is convex, substituting the
definition of $\gamma_{k}^\mathrm{ID}$ into \eqref{B.2} leads to a difference of two convex functions (DC decomposition) and thus it is non-convex.
The following DC decomposition is thus applied to \eqref{B.2}:
\begin{align}
\mathcal{M}_1(\mathbf{z})-\mathcal{N}_1(\mathbf{z}) \leq 0, \label{B.3}
\end{align}
where $\mathbf{z} = [\mathbf{p}_{k}, \mathbf{f}_{j}, \rho_k],\ \forall  k \in \mathcal{K}, j\in \mathcal{J}$,
$\mathcal{N}_1(\mathbf{z}) \triangleq \frac{\left|\mathbf{h}_{t,k}\mathbf{p}_{k}\right|^2}{\rho_k}$, and $\mathcal{M}_1(\mathbf{z}) \triangleq \sum_{k' \in \mathcal{K},k' \neq k}\left|\mathbf{h}_{t,k}\mathbf{p}_{k'}\right|^2+\sum_{j \in \mathcal{J}}\left|\mathbf{h}_{t,k}\mathbf{f}_{j}\right|^2+1$.
The non-convexity of \eqref{B.2} is caused by $\mathcal{N}_1(\mathbf{z})$, as shown in \eqref{B.3}, and thus the non-convex factor is replaced at the $i$-th iteration by a first order Taylor expansion of $\mathcal{N}_1(\mathbf{z})$ using the
operator $\Psi^{[i]}(\mathbf{u},x;\mathbf{h})$.

\section{Proof of Proposition 5} \label{app3}
Due to the nonlinear term $\alpha_k R_c$, \eqref{P1-e} is non-convex. To deal with this non-convexity, we
first write the equivalent DC decomposition of \eqref{P1-e} as:
\begin{align}
r_c + (\mathcal{A}(\alpha_k,R_c)-\mathcal{B}(\alpha_k,R_c))  \leq 0, \label{C.1}
\end{align}
where $\mathcal{A}(\alpha_k,R_c) \triangleq \frac{1}{4} (\alpha_k-R_c)^2$ and $\mathcal{B}(\alpha_k,R_c) \triangleq \frac{1}{4} (\alpha_k+R_c)^2$. As  can be observed, the non-convexity of \eqref{C.1} is
caused by $\mathcal{B}(\alpha_k,R_c)$. Therefore, we replace it by the affine
approximation obtained by the first-order Taylor expansion at
the $i$-th iteration, i.e., $\hat{\mathcal{B}}^{[i]}(\alpha_k,R_c)$. Subsequently, inserting
$\hat{\mathcal{B}}^{[i]}(\alpha_k,R_c)$ into \eqref{C.1}, the affine approximation of \eqref{P1-e}
is given by:
\begin{align}
\Theta^{[i]}(\alpha_k,R_c) \geq r_c. 
\end{align}

\section{Proof of Proposition 6} \label{app4}
The constraint \eqref{P1-f} represents a search over all possible values of the channel uncertainty to find the worst-case.
Using the definition of $Q_j$ from \eqref{Q-j def} and   some variable transformations, four extra auxiliary constraints emerge, as follows:
\begin{align}
&\displaystyle\sum_{j \in \mathcal{J}} \left(\lambda_{j,c} + \displaystyle\sum_{k \in \mathcal{K}} \lambda_{j,k} + \displaystyle\sum_{j' \in \mathcal{J}} \xi_{j,j'}
\right) \geq E^{th},
\\
&\min_{ \mathbf{\Delta g}_{r,j}} \left|\mathbf{h}_{r,j}\mathbf{p}_{c}\right|^2 \geq \lambda_{j,c},  \ \forall j \in \mathcal{J},\label{D.3}
\\
&\min_{ \mathbf{\Delta g}_{r,j}} \left|\mathbf{h}_{r,j}\mathbf{p}_{k}\right|^2 \geq \lambda_{j,k}, \ \forall j \in \mathcal{J}, k \in \mathcal{K},   \label{D.4}
\\
&\min_{ \mathbf{\Delta g}_{r,j}} \left|\mathbf{h}_{r,j}\mathbf{f}_{j'}\right|^2  \geq \xi_{j,j'},
\ \forall j,j' \in \mathcal{J}. \label{D.5}
\end{align}
It can be observed that \eqref{D.3}-\eqref{D.5} are still non-convex. Utilizing proposition 1 in Appendix \ref{app1}, the convex counterparts can be obtained as shown in \eqref{lem6-eq}.

\begin{IEEEbiography}[{\includegraphics[width=1in,height=1.25in,clip,keepaspectratio]{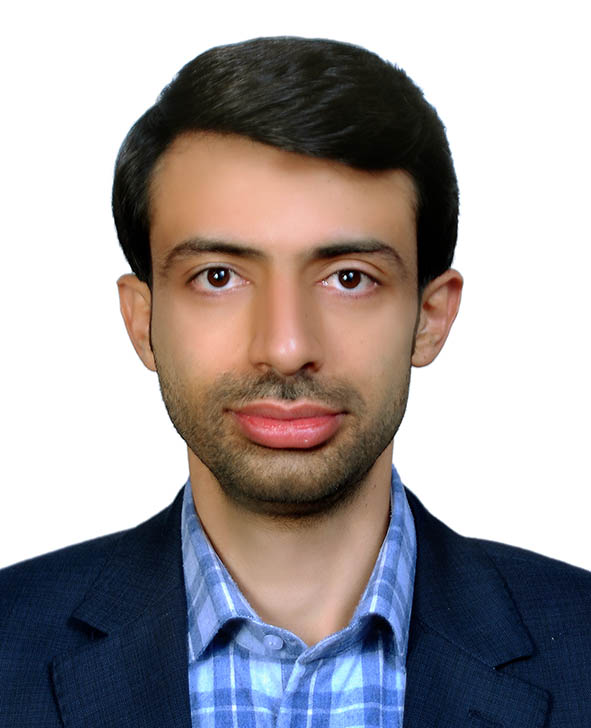}}]{Hamid Reza Hashempour}
 was born in 1987 and obtained his B.S., M.S., and Ph.D. degrees in Electrical Engineering from Shiraz University, Shiraz, Iran, in 2009, 2011, and 2017, respectively. Currently, he is a Postdoctoral Fellow at the Department of Electronic Systems at Aalborg University, Aalborg, Denmark.
	His research interests encompass wireless communication, physical-layer security, federated learning and 6G in-X subnetworks. In recognition of his work, his paper presented at the ICEE2017 conference was awarded the Best Paper of the Conference.
\end{IEEEbiography}
\begin{IEEEbiography}[{\includegraphics[width=1in,height=1.25in,clip,keepaspectratio]{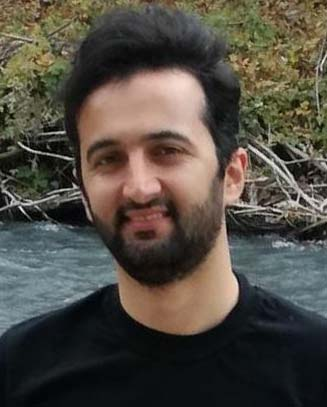}}]{Hamed Bastami}
	received the M.Sc. degree in electrical engineering from the University of Tehran,
	Tehran, Iran, in 2014. He is currently pursuing
	the Ph.D. degree with the Department of Electrical
	Engineering, Sharif University of Technology. His
	research interests lie in the areas of physical-layer
	security (PLS) of wireless communications with
	special emphasis on machine learning, deep-based
	resource management, PLS in wireless communication, and fog radio access networks.
\end{IEEEbiography}
\begin{IEEEbiography}[{\includegraphics[width=1in,height=1.25in,clip,keepaspectratio]{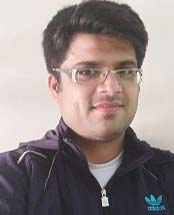}}]{Majid Moradikia}
is currently a Research Scientist at Worcester Polytechnic Institute (WPI). Hereceived his Ph.D. degree in Electrical Engineeringfrom the Department of Electrical and ElectronicsEngineering, Shiraz University of Technology, Shiraz. Currently at WPI, he serves as part of theSoilX project team, a cutting-edge project sponsored by the US Department of Agriculture (USDA)to develop an intelligent radar technology for soilmoisture characterization. Prior to joining WPI, hehad a Research Scholar position at the Departmentof Electrical Engineering, Sharif University of Technology, Tehran, and aPostdoctoral Fellow at the University of Houston (UH). In UH, Dr. Moradikiacontributed to a team that won first place of the 4th Annual Beyond 5GSoftware Defined Radio University Challenge held by Airforce ResearchLaboratory (AFRL). His research interests include wireless communications,security and privacy, Machine Learning, and Radar and CommunicationsSignal Processing.
\end{IEEEbiography}
\begin{IEEEbiography}[{\includegraphics[width=1in,height=1.25in,clip,keepaspectratio]{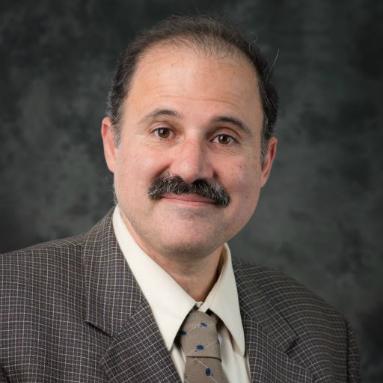}}]{Seyed A. (Reza) Zekavat}
	 is currently a Professor of Data Sciences and Physics with the Worcester Polytechnic Institute. He is also the author of the textbook Electrical Engineering: Concepts and Applications (Pearson), and an Editor of the book Handbook of Position Location: Theory, Practice and Advances, (Wiley/IEEE). He holds a patent on an active Wireless Remote Positioning System. He has also coauthored two books Multi-Carrier Technologies for Wireless Communications (Kluwer), and High Dimensional Data Analysis (VDM Verlag); and 20 book chapters in the areas of adaptive antennas, localization, and spectrum sharing. His research interests include wireless communications, positioning systems, software defined radio design, dynamic spectrum allocation methods, radar theory, blind signal separation and MIMO and beam forming techniques, feature extraction, and neural networking. He is active on the technical program committees for several IEEE international conferences, serving as the Committee Chair or a member. He has published more than 150 articles on these topics. He served on the Editorial Board of many Journals, including IET Communications, IET Wireless Sensor Systems, and Springer International Journal on Wireless Networks. He has been on the Executive Committee of multiple IEEE conferences.
\end{IEEEbiography}
\begin{IEEEbiography}[{\includegraphics[width=1in,height=1.25in,clip,keepaspectratio]{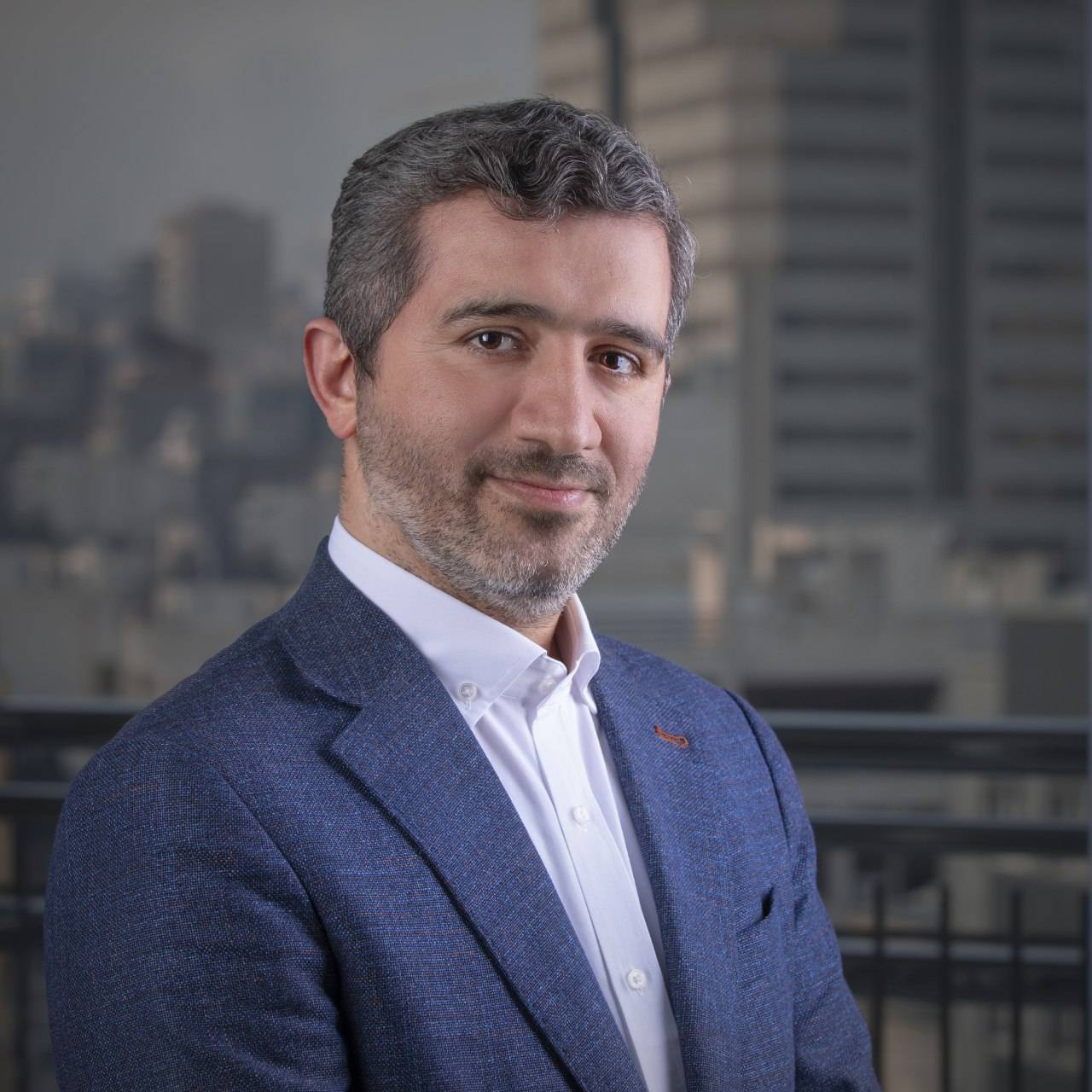}}]{Hamid Behroozi}
	 received the B.Sc. degree in electrical engineering from the University of Tehran, Tehran, Iran, in 2000, the M.Sc. degree in electrical engineering from the Sharif University of Technology, Tehran, in 2003, and the Ph.D. degree in electrical engineering from Concordia University, Montreal, QC, Canada, in 2007. After his Ph.D. degree, he was a Post-Doctoral Research Fellow with the Department of Mathematics and Statistics, Queen’s University, Kingston, ON, Canada, from 2007 to 2010. He is currently an Associate Professor with the Department of Electrical Engineering, Sharif University of Technology. His current research interests include communication and information theory, machine learning, and 6G and beyond mobile communications with a special emphasis on semantic communications, information security and privacy of networked systems. Dr. Behroozi was a recipient of several academic awards, including the Ontario Postdoctoral Fellowship awarded by the Ontario Ministry of Research and Innovation (MRI), the Quebec Doctoral Research Scholarship awarded by the Government of Quebec (FQRNT), the Hydro Quebec Graduate Award, and the Concordia University Graduate Fellowship.
\end{IEEEbiography}
\begin{IEEEbiography}[{\includegraphics[width=1in,height=1.25in,clip,keepaspectratio]{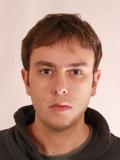}}]{Gilberto Berardinelli}
	received the first and second level degrees (cum laude) in telecommunication engineering from the University of L’Aquila, Italy, in 2003 and 2005, respectively, and the Ph.D. degree from Aalborg University, Denmark, in 2010. He is currently an Associate Professor with the Wireless Communication Networks (WCN) Section, Aalborg University, and external research engineer at Nokia Bell Labs. He has also been involved in multiple European projects, such as FANTASTIC5G, ONE5G, 5GSmartFact, and he is currently coordinator and technical manager of the Horizon Europe 6G-SHINE project, focused on pioneering technology components for short-range wireless communication with extreme requirements. He is author or co-author of more than 150 publications including journals, conferences, book chapters and patents applications. His current research interests are mostly focused on medium access control and radio resource management design for 6G systems and beyond. He is an IEEE Senior Member, and serves as a technical editor for IEEE Wireless Communications magazine.
\end{IEEEbiography}
\begin{IEEEbiography}[{\includegraphics[width=1in,height=1.25in,clip,keepaspectratio]{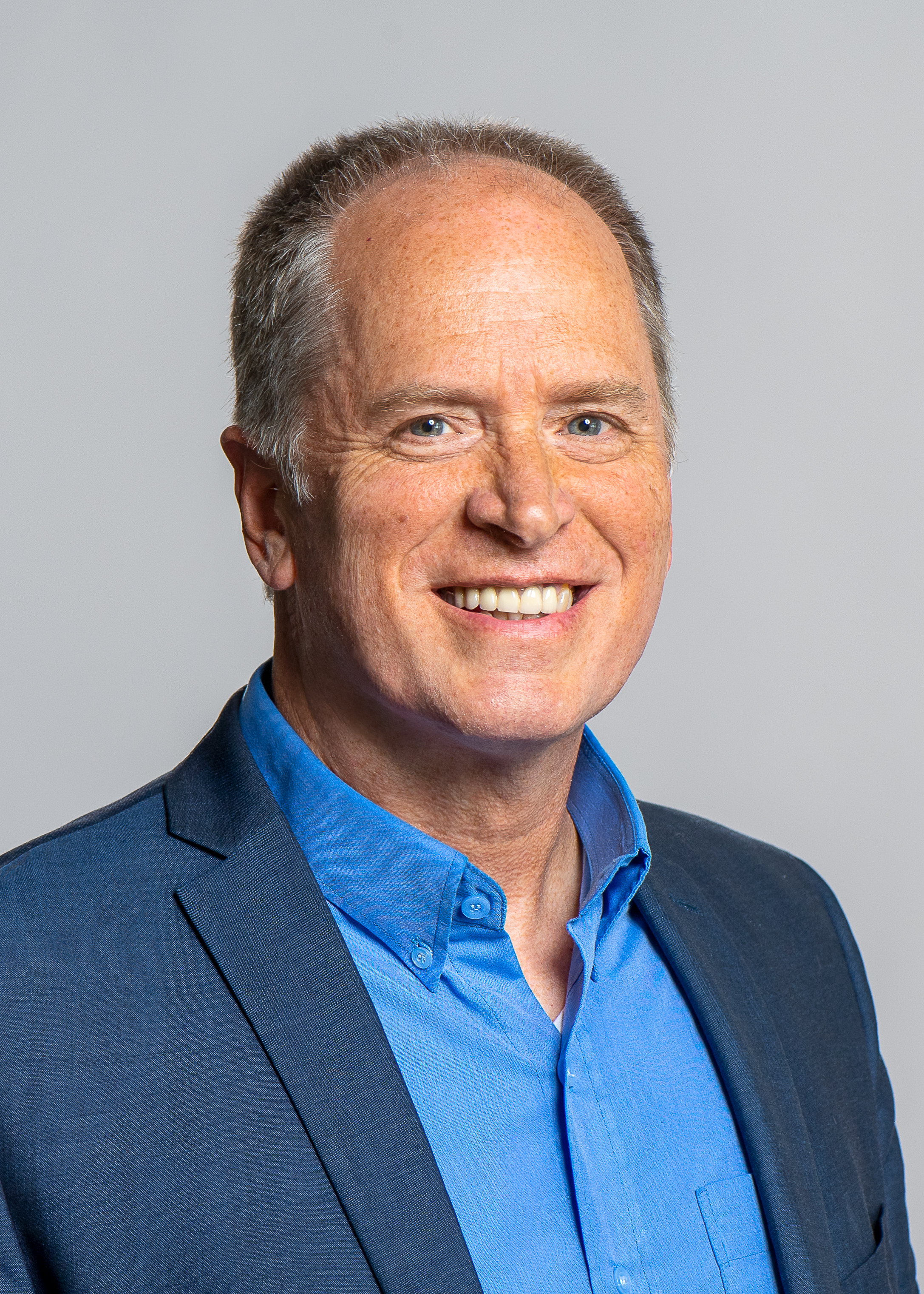}}]{Lee Swindlehurst}
 received the B.S. (1985) and M.S. (1986) degrees in Electrical Engineering from Brigham Young University (BYU), and the PhD (1991) degree in Electrical Engineering from Stanford University. He was with the Department of Electrical and Computer Engineering at BYU from 1990-2007, and during 1996-97, he held a joint appointment as a visiting scholar at Uppsala University and the Royal Institute of Technology in Sweden. From 2006-07, he was on leave working as Vice President of Research for ArrayComm LLC in San Jose, California. Since 2007 he has been with the Electrical Engineering and Computer Science (EECS) Department at the University of California Irvine, where he is currently a Distinguished Professor and holds the Nicolaos G. and Sue Curtis Alexopoulos Presidential Chair. During 2014-17 he was also a Hans Fischer Senior Fellow in the Institute for Advanced Studies at the Technical University of Munich. In 2016, he was elected as a Foreign Member of the Royal Swedish Academy of Engineering Sciences (IVA). His research focuses on array signal processing for radar, wireless communications, and biomedical applications. Dr. Swindlehurst is a Fellow of the IEEE and was the inaugural Editor-in-Chief of the IEEE Journal of Selected Topics in Signal Processing. He received the 2000 IEEE W. R. G. Baker Prize Paper Award, the 2006 IEEE Communications Society Stephen O. Rice Prize in the Field of Communication Theory, the 2006, 2010 and 2021 IEEE Signal Processing Society’s Best Paper Awards, the 2017 IEEE Signal Processing Society Donald G. Fink Overview Paper Award, a Best Paper award at the 2020 IEEE International Conference on Communications, and the 2022 Claude Shannon-Harry Nyquist Technical Achievement Award from the IEEE Signal Processing Society.
\end{IEEEbiography}
\end{document}